\documentclass[journal]{IEEEtran}
\usepackage{amsmath,amssymb,amsfonts,bm}
\usepackage{graphicx}
\usepackage{algorithm}
\usepackage{algorithmic}

\newcommand{\vect}[1]{\mathbf{#1}}
\def\diag{\mathrm{diag}}
\def\tr{\mathrm{tr}}
\def\Htran{\mbox{\tiny $\mathrm{H}$}}
\def\Ttran{\mbox{\tiny $\mathrm{T}$}}
\def\CN{\mathcal{N}_{\mathbb{C}}}

\newcommand{\E}{\mathbb{E}}
\newcommand{\blkdiag}{\mathrm{blkdiag}}
\newtheorem{remark}{Remark}
\newtheorem{lemma}{Lemma}
\newtheorem{proposition}{Proposition}

\begin{document}
\title{Energy-Efficient Integrated Access and Fronthaul for Cell-Free Massive MIMO with Adaptive Quantization Resolution}
\author{\"Ozlem Tu\u{g}fe Demir  
 
\thanks{  \"O. T. Demir is with the Department of Electrical and Electronics Engineering, Bilkent University, Ankara, T\"urkiye (ozlemtugfedemir@bilkent.edu.tr).    \newline \indent
This work was supported by The Scientific and Technological Research Council of Türkiye (TUBITAK) BIDEB
2232-B International Fellowship for Early Stage Researchers under Grant Number 122C149. This work was supported by the BAGEP Award of the
Science Academy //with funding supplied by Özgül ve Armağan Çağlayan Vakfı.// 
}\vspace{-7mm}
}
\maketitle

\begin{abstract}
Cell-free massive MIMO with wireless fronthaul is a promising architecture for
energy-efficient 6G networks, but the access and fronthaul links must then
share the same scarce spectrum, and, under the fully centralized (option-8)
functional split, the fronthaul rate is dictated by the finite quantization
resolution used at the access points (APs). This paper develops a network
energy-efficiency (EE) maximization framework for the uplink of such a system,
jointly optimizing the integrated access and fronthaul (IAF) resource split,
the adaptive per-AP quantization resolution, and the fronthaul powers, and
treating the time-division (TD) and frequency-division (FD) operating modes in
a unified manner. Each AP may be switched off (put to sleep) when it is not
worth activating, so the resolution allocation is inherently coupled with AP
selection. The resulting mixed-integer, nonconvex fractional program is solved
by an alternating-optimization algorithm with per-block optimality
guarantees---a closed-form optimal time split, bandwidth bisection, and optimal
per-AP bit selection---that applies verbatim to both modes. While the design
relies on the tractable additive quantization noise model, the reported
performance is obtained end-to-end with the actual Lloyd--Max quantizers and a
Bussgang decomposition-based achievable-rate bound. 
\end{abstract}

\vspace{-4mm}
\section{Introduction}

Cell-free massive multiple-input multiple-output (MIMO) has emerged as a
key enabler for beyond-5G and 6G networks, owing to its ability to deliver
uniformly high data rates across the entire coverage area
\cite{ngo2017cellfree,demir2021foundations,chowdhury2020_6g}. By distributing a large number
of access points (APs) over a wide area and having them jointly serve the
user equipments (UEs) under the coordination of a central processing unit
(CPU), cell-free massive MIMO removes the cell boundaries of conventional
co-located deployments and provides macro-diversity, so that even UEs in
unfavorable propagation conditions enjoy reliable service
\cite{ngo2017cellfree,demir2021foundations}. This architecture is naturally
aligned with the disaggregated Open RAN vision, in which the physical-layer
processing is functionally split between the distributed radio units (the
APs) and a centralized baseband unit (the CPU) \cite{topal2026unlocking}.

Realizing this vision at scale hinges on the fronthaul that connects the APs
to the CPU. Under the fully centralized functional split (option~8), each AP
forwards the digitized in-phase/quadrature (I/Q) samples of all its antennas
to the CPU, where channel estimation and receive combining are performed
\cite{topal2026unlocking}. Provisioning a dense layer of APs with dedicated
optical fronthaul is, however, costly and inflexible, which motivates
\emph{wireless} fronthaul as a scalable and rapidly deployable alternative
\cite{topal2026unlocking,madapatha2020iab,topal2024eewireless,demir2026hwimpaired,demirhan2022enabling}.
A range of wireless- and capacity-limited fronthaul designs has recently been
investigated by several groups, including mmWave wireless fronthaul
\cite{demirhan2022enabling}, capacity-constrained fronthaul
\cite{femenias2019cellfree}, amplify-and-forward relaying under hardware
impairments \cite{demir2026hwimpaired}, and over-the-air fronthaul signaling
that computes the CPU's sufficient statistics directly over the air
\cite{shaik2025ota}.
Wireless fronthaul, in turn, must
share the scarce radio spectrum with the access links and delivers a rate
that varies with the allocated bandwidth and the instantaneous link
conditions. This tight coupling gives rise to an \emph{integrated
access and fronthaul} (IAF) problem, in which the radio resources---time,
bandwidth, and transmit power---must be jointly apportioned between the
access and fronthaul links, in the same spirit as integrated access and
backhaul (IAB) for cellular networks
\cite{madapatha2020iab}.

The IAF architecture studied here is a fronthaul-domain counterpart of
IAB, a concept standardized by 3GPP
(Rel-16/17) and extensively studied as a means of densifying cellular
networks without provisioning wired backhaul to every node
\cite{madapatha2020iab,polese2020iab}. The core idea of IAB is to reuse the
same radio spectrum and hardware for the access and backhaul links, so that a
node relays its traffic wirelessly toward a fiber-connected donor; this
sharply lowers deployment cost and enables rapid, flexible densification, at
the price of having to jointly schedule and allocate the shared
time--frequency--power resources between the two link types
\cite{polese2020iab}. These pressures are even more acute in cell-free massive
MIMO, where a large number of APs must be connected to the CPU and equipping
every AP with fiber is neither cost-effective nor scalable. This has motivated
recent extensions of the IAB principle to cell-free systems, in which the
access and wireless fronthaul/backhaul links share the spectrum and the
bandwidth split and beamformers are optimized to maximize the end-to-end rate
\cite{jazi2023iab}. However, such works target the spectral efficiency and do not
address the finite-resolution fronthaul quantization, the adaptive per-AP bit
allocation, or the network energy efficiency (EE) that are central to this paper.

A distinctive feature of the split-option-8 fronthaul is that the AP output
is \emph{quantized} to a finite resolution before being conveyed to the CPU,
so that the fronthaul rate is governed directly by the number of
quantization bits per sample \cite{masoumi2020performance,khorsandmanesh2023optimized}.
Coarser quantization lowers the fronthaul load at the expense of additional
distortion, creating a fundamental resolution--distortion trade-off that has
been investigated for multi-user MIMO \cite{khorsandmanesh2023optimized}
and for cell-free systems with limited-capacity fronthaul
\cite{masoumi2020performance,kim2024meta}. Because the wireless-fronthaul
rate fluctuates with the allocated resources, the number of bits that each AP
can afford to transmit is not fixed but must itself be optimized: assigning
the available bits \emph{unevenly} across APs, according to their relative
importance, uses the fronthaul far more efficiently than uniform allocation
\cite{kim2024meta,demir2026quantization}. To render such bit-allocation
designs analytically tractable, the quantization distortion is commonly
abstracted through the additive quantization noise model (AQNM)
\cite{fletcher2007robust} or, more accurately, through the Bussgang
decomposition \cite{demir2020bussgang}.

These considerations become especially pressing in wide-band deployments such
as the emerging upper mid-band  spectrum \cite{bazzi2025upper}, where hundreds of
megahertz of bandwidth are available but the thermal-noise power grows in
proportion to the utilized bandwidth. On the one hand, a
larger bandwidth raises the pre-log factor of the achievable rate; on the
other hand, it raises the noise floor and the bandwidth-proportional circuit
power, and it forces the fronthaul to convey far more samples per second. In
this regime, operating each AP at a \emph{low} quantization resolution---and
switching lightly loaded APs off altogether---is frequently the
energy-optimal choice, which makes joint bandwidth, bit, and sleep-mode
control indispensable for the network EE
\cite{lopezperez2021energy,enqvist2024fundamentals,enqvist2026sleep}.
Indeed, EE has become a first-order design objective for
sustainable 6G networks \cite{chowdhury2020_6g}, and a growing body of work
optimizes the EE of cell-free massive MIMO with wireless
fronthaul through access-point activation, sleep modes, and joint
radio--fronthaul--cloud orchestration
\cite{topal2026unlocking,topal2024eewireless,demir2024oran}.
Recent 3GPP discussions on next-generation MIMO for 6G point in the same
direction, emphasizing the reduction of power consumption and signaling
overhead by muting or deactivating antenna ports, as reflected in technical
documents from the 3GPP 6G Workshop held in March~2025
\cite{Nokia_6GWS,Ericsson_6GWS}; these reports mark a shift in the key
performance indicators from a sole focus on spectral efficiency (raw
throughput) toward EE \cite{kosasih_balkancom}.

Despite the rich literature on cell-free massive MIMO and on fronthaul
quantization, the \emph{joint} design of the IAF resource split and the
adaptive fronthaul resolution for a wireless-fronthaul cell-free system
operating over wide bands has, to the best of our knowledge, not been
addressed. Existing works either fix the fronthaul rate or bit budget and
optimize only the quantization/precoding
\cite{masoumi2020performance,kim2024meta},
or orchestrate the radio resources without a distortion-aware
EE objective \cite{topal2026unlocking}. In this paper, we
bridge this gap. Our main contributions are summarized as follows.

\begin{itemize}
  \item \textbf{Unified IAF framework.} We formulate the uplink of a
  wireless-fronthaul cell-free massive MIMO network in which the access and
  fronthaul links share a common band, and we treat the time-division (TD)
  and frequency-division (FD) operating modes within a single optimization
  framework, exposing their complementary time-versus-bandwidth trade-offs.

  \item \textbf{Adaptive per-AP quantization with sleep modes.} Each AP is
  assigned an adaptive data resolution $b_l \in \{0,1,\dots,b_{\max}\}$,
  where $b_l = 0$ places the AP into sleep, and these resolutions are jointly
  optimized together with the access/fronthaul time split, the per-link
  bandwidths, and the fronthaul transmit powers under a realistic power model
  that captures bandwidth- and rate-proportional processing as well as
  sleep-mode savings.

  \item \textbf{EE maximization algorithm.} The resulting
  network EE-maximization is a mixed-integer, nonconvex fractional program;
  we develop an alternating-optimization algorithm with per-block optimality
  guarantees---including a closed-form optimal time split, bandwidth
  bisection, and optimal per-AP bit selection---that applies verbatim to both
  the TD and FD modes.

  \item \textbf{Distortion-aware, end-to-end validation.} While the design
  relies on the tractable AQNM statistics, we assess the delivered
  performance in an end-to-end manner using the actual Lloyd--Max quantizers
  and a Bussgang decomposition-based (use-and-then-forget) achievable-rate bound, thereby
  accounting for the true, input-dependent quantization errors beyond the
  AQNM assumptions.

\end{itemize}

\vspace{-2mm}

\section{System Model and Problem Formulation}\label{sec:model}

\subsection{Network and Integrated Access--Fronthaul Frame Structure}
\label{subsec:frame}
We consider the uplink of a cell-free massive MIMO network in which $L$ APs, each equipped with $N$ antennas, jointly serve $K$ single-antenna UEs under centralized operation. The total number of AP antennas is $M = LN$. The APs are connected to the CPU via \emph{wireless} fronthaul links, and functional split option~8 is adopted: each AP forwards the digitized (quantized) baseband I/Q samples of all its antennas to the CPU, where channel estimation and receive combining are carried out \cite{topal2026unlocking}.

The access and fronthaul links share the same total bandwidth $B$, and we study and compare two integrated access--fronthaul operating modes that differ in how the radio resources are shared between the two links. Each resource-allocation frame has (normalized) unit duration, and $t_1$ and $t_2$ denote the fractions of the frame during which the access and fronthaul links are active, respectively:
\begin{itemize}
    \item \emph{Time-division (TD) mode:} the two links are separated in time to avoid cross-link interference. The frame is divided into an \emph{access phase} of duration $t_1$, during which the UEs transmit uplink pilot and data signals to the APs over a bandwidth $B_1 \leq B$, and a \emph{fronthaul phase} of duration $t_2 = 1 - t_1$, during which the APs forward their buffered quantized I/Q samples to the CPU over a bandwidth $B_2 \leq B$. Since the two phases never overlap in time, both $B_1$ and $B_2$ can individually be as large as the full band, and the frame is fully utilized, i.e., $t_1 + t_2 = 1$.
    \item \emph{Frequency-division (FD) mode:} the band is partitioned into two non-overlapping subbands with $B_1 + B_2 \leq B$, so that the access and fronthaul links operate \emph{simultaneously} without cross-link interference. The activity factors are then decoupled: the access link is active during a fraction $t_1 \leq 1$ of the frame and the fronthaul link during a fraction $t_2 \leq 1$, while each radio interface is put to sleep for the remaining time to save energy.
\end{itemize}
In both modes, the noise power scales with the utilized bandwidth, i.e., the noise variances over the access and fronthaul links are
\begin{equation}
    \sigma_1^2 = N_0 B_1, \qquad \sigma_2^2 = N_0 B_2,
\end{equation}
where $N_0$ is the noise power spectral density (including the noise figure). The two modes exhibit complementary trade-offs that are investigated throughout the paper: in the TD mode, each link may exploit the entire band but the links compete for transmission time, whereas in the FD mode, each link may transmit during the entire frame but the links compete for bandwidth, and shrinking the activity factors saves energy through sleep at the expense of throughput and fronthaul capacity. In both modes, enlarging $B_1$ (or $B_2$) increases the pre-log factor of the corresponding rate but raises the noise floor and the bandwidth-dependent circuit power.

\vspace{-2mm}
\subsection{Uplink Access Phase}
The time-frequency resources of the access phase are divided into coherence blocks of $\tau_c$ symbols, of which $\tau_p$ symbols are used for uplink pilots and $\tau_u = \tau_c - \tau_p$ for uplink data \cite{demir2021foundations}. The channel between UE $k$ and AP $l$ is denoted by $\vect{h}_{lk} \in \mathbb{C}^{N}$ and is modeled as correlated Rayleigh fading,
\begin{equation}
    \vect{h}_{lk} \sim \CN\left(\vect{0}, \vect{R}_{lk}\right),
\end{equation}
where $\vect{R}_{lk} \in \mathbb{C}^{N \times N}$ is the spatial correlation matrix and $\beta_{lk} = \tr(\vect{R}_{lk})/N$ is the large-scale fading coefficient. The collective channel of UE $k$ is $\vect{h}_k = [\vect{h}_{1k}^{\Ttran} \, \cdots \, \vect{h}_{Lk}^{\Ttran}]^{\Ttran} \in \mathbb{C}^{M}$.

During the uplink data transmission, the received signal at AP $l$ is
\begin{equation}
    \vect{y}_l = \sum_{i=1}^{K} \sqrt{p_i}\, \vect{h}_{li} s_i + \vect{n}_l \in \mathbb{C}^{N},
    \label{eq:received_signal}
\end{equation}
where $s_i$ is the unit-power data symbol of UE $i$, $p_i$ is its transmit power---fixed at the maximum value $p_{\max}$ throughout this work (see Remark~\ref{rem:uepower})---and $\vect{n}_l \sim \CN(\vect{0}, \sigma_1^2 \vect{I}_N)$ is the additive noise with the bandwidth-dependent variance $\sigma_1^2 = N_0 B_1$.

\subsection{Fronthaul Quantization With Adaptive Per-AP Resolution}
\label{subsec:quantization}
Under split option~8, AP $l$ quantizes each real dimension (I and Q) of every antenna's baseband samples and buffers the resulting bit stream for transmission during the fronthaul phase. To decouple the channel estimation quality from the data-quantization design, the pilot samples are quantized with a \emph{fixed} resolution of $b_{\rm p}$ bits (common to all APs), whereas the data samples are quantized with an \emph{adaptive} per-AP resolution of $b_l$ bits, which is a key optimization variable of this work. The data resolutions take values $b_l \in \{0, 1, \ldots, b_{\max}\}$, where $b_l = 0$ indicates that AP $l$ is \emph{switched off} (put to sleep) for the entire frame: it neither receives during the access phase nor forwards any samples, and it is excluded from the centralized processing and from the fronthaul transmission, consuming only the sleep power specified in Section~\ref{subsec:power_model}. We denote the set of active APs by $\mathcal{A} = \{ l : b_l \geq 1 \}$ and its cardinality by $L_{\mathcal{A}} = |\mathcal{A}|$. For the analytical modeling and optimization, we adopt the AQNM \cite{fletcher2007robust}, under which the quantized data signal at AP $l$ is
\begin{equation}
    \check{\vect{y}}_l = \left(1 - \eta_l\right) \vect{y}_l + \vect{q}_l,
    \label{eq:aqnm}
\end{equation}
where $\eta_l \in (0,1)$ is the distortion factor of AP $l$ and the quantization distortion $\vect{q}_l$ is uncorrelated with $\vect{y}_l$ and has the diagonal covariance matrix
\begin{equation}
    \E\left\{\vect{q}_l \vect{q}_l^{\Htran}\right\} = \eta_l \left(1 - \eta_l\right) \diag\left(\vect{R}_{\vect{y},l}\right),
\end{equation}
with $\vect{R}_{\vect{y},l} = \E\{\vect{y}_l \vect{y}_l^{\Htran}\}=\sum_{i=1}^Kp_i\vect{R}_{li}+\sigma_1^2\vect{I}_N$. For a Lloyd--Max scalar quantizer with resolution $b_l > 5$ bits, the distortion factor is accurately approximated by high-rate quantization theory as \cite{gray1998quantization}
\begin{equation}
    \eta_l \approx c_q\, 2^{-2 b_l}, \qquad c_q = \frac{\sqrt{3}\pi}{2},
    \label{eq:distortion_factor}
\end{equation}
while for $b_l \leq 5$ the exact tabulated Lloyd--Max distortion factors are used \cite[Tab.~4.1]{jain1989fundamentals}. After compensating for the deterministic scaling $(1-\eta_l)$ at the CPU, the effective signal available for centralized processing is
\begin{equation}
    \tilde{\vect{y}}_l = \frac{\check{\vect{y}}_l}{1 - \eta_l} = \vect{y}_l + \vect{e}_l, \qquad
    \vect{e}_l = \frac{\vect{q}_l}{1 - \eta_l},
\end{equation}
where the normalized distortion $\vect{e}_l$ has the covariance
\begin{equation}
    \vect{C}_{\vect{e},l} = \frac{\eta_l}{1 - \eta_l}\, \diag\left(\vect{R}_{\vect{y},l}\right).
    \label{eq:distortion_cov}
\end{equation}

\emph{Fronthaul load:} Since the sampling rate is proportional to the access bandwidth $B_1$ and two real samples are generated per complex baseband sample per antenna, and since split option~8 forwards the raw antenna-domain observations of both the pilot and data portions of each coherence block, the number of bits that AP $l$ must deliver to the CPU per frame is
\begin{equation}
    F_l = 2 N B_1 t_1 \left( \frac{\tau_p}{\tau_c} b_{\rm p} + \frac{\tau_u}{\tau_c} b_l \right) \quad \text{[bits/frame]},
    \label{eq:fronthaul_load}
\end{equation}
where the first term is the fixed pilot-forwarding overhead and the second term is the adaptive data load.

\subsection{Quantization-Aware Channel Estimation}
\label{subsec:channel_estimation}
Since the pilot observations also traverse the fronthaul quantizers, channel estimation at the CPU is affected by the quantization distortion. During the pilot phase, UE $k$ transmits a pilot sequence $\bm{\phi}_k \in \mathbb{C}^{\tau_p}$ with $\|\bm{\phi}_k\|^2 = \tau_p$, and we let $\mathcal{P}_k \subset \{1,\ldots,K\}$ denote the set of UEs sharing the pilot of UE $k$ (including itself). After descaling and correlating the quantized pilot signal with $\bm{\phi}_k$, the CPU obtains
\begin{equation}
    \tilde{\vect{z}}_{lk} = \sum_{i \in \mathcal{P}_k} \sqrt{p_i \tau_p} \vect{h}_{li} + \tilde{\vect{n}}_{lk} + \tilde{\vect{e}}_{lk},
\end{equation}
where $\tilde{\vect{n}}_{lk} \sim \CN(\vect{0}, \sigma_1^2 \vect{I}_N)$ and $\tilde{\vect{e}}_{lk}$ is the accumulated quantization distortion, which under AQNM is uncorrelated across pilot samples and has the covariance
\begin{equation}
    \vect{C}_{\vect{e},l}^{\rm p} = \frac{\eta_{\rm p}}{1-\eta_{\rm p}}\, \diag\Bigg(\sum_{i=1}^{K} p_i \vect{R}_{li} + \sigma_1^2 \vect{I}_N\Bigg)
\end{equation}
computed from the pilot-phase received-signal statistics, where $\eta_{\rm p} \approx c_q 2^{-2 b_{\rm p}}$ is the distortion factor of the fixed pilot resolution.\footnote{Since the pilot resolution $b_{\rm p}$ is fixed and common to all APs, the channel estimates and their error statistics do not depend on the optimization variables $\{b_l\}$, which keeps the estimation stage decoupled from the adaptive data-quantization design. Under the AQNM, the quantization distortion is modeled as temporally and spatially uncorrelated, which renders the LMMSE estimator below tractable; the temporal correlation of the true quantization errors over the pilot sequence is captured in the numerical evaluation through the end-to-end Bussgang-based processing described in Remark~\ref{rem:bussgang}.} The quantization-aware linear minimum mean-squared error (LMMSE) estimate of $\vect{h}_{lk}$ is then
\begin{equation}
    \hat{\vect{h}}_{lk} = \sqrt{p_k\tau_p}\, \vect{R}_{lk} \bm{\Psi}_{lk}^{-1} \tilde{\vect{z}}_{lk},
    \label{eq:lmmse}
\end{equation}
where
\begin{equation}
    \bm{\Psi}_{lk} = \E\left\{\tilde{\vect{z}}_{lk}\tilde{\vect{z}}_{lk}^{\Htran}\right\} = \sum_{i \in \mathcal{P}_k} p_i \tau_p\, \vect{R}_{li} + \sigma_1^2 \vect{I}_N + \vect{C}_{\vect{e},l}^{\rm p}.
\end{equation}
The estimate and the estimation error $\tilde{\vect{h}}_{lk} = \vect{h}_{lk} - \hat{\vect{h}}_{lk}$ are uncorrelated, with error covariance
\begin{equation}
    \vect{C}_{lk} = \vect{R}_{lk} - p_k \tau_p\, \vect{R}_{lk} \bm{\Psi}_{lk}^{-1} \vect{R}_{lk}.
\end{equation}

\subsection{Centralized Combining and Achievable Spectral Efficiency}
\label{subsec:rate}
Collecting the descaled quantized observations of all APs in $\tilde{\vect{y}} = [\tilde{\vect{y}}_1^{\Ttran} \, \cdots \, \tilde{\vect{y}}_L^{\Ttran}]^{\Ttran} \in \mathbb{C}^{M}$, the CPU applies the receive combining vector $\vect{v}_k \in \mathbb{C}^{M}$ to detect the data of UE $k$: $\hat{s}_k = \vect{v}_k^{\Htran} \tilde{\vect{y}}$. The combining vectors follow the convention that the entries corresponding to switched-off APs are set to zero, so that only the observations of the active APs contribute to the detection. Following the standard centralized-operation bound for cell-free massive MIMO \cite[Sec.~5]{demir2021foundations}, extended to include the quantization distortion, an achievable spectral efficiency (SE) of UE $k$ is
\begin{equation}
    \mathsf{SE}_k = \frac{\tau_u}{\tau_c}\, \E\left\{ \log_2\left(1 + \mathsf{SINR}_k \right)\right\} \quad \text{[bit/s/Hz]},
    \label{eq:SE}
\end{equation}
where the expectation is with respect to the channel estimates and the effective instantaneous signal-to-interference-plus-noise ratio (SINR) is
\begin{equation}
    \mathsf{SINR}_k = \frac{p_k \left| \vect{v}_k^{\Htran} \hat{\vect{h}}_k \right|^2}{\sum\limits_{i=1, i \neq k}^{K} p_i \left| \vect{v}_k^{\Htran} \hat{\vect{h}}_i \right|^2 + \vect{v}_k^{\Htran} \vect{Z} \vect{v}_k},
    \label{eq:SINR}
\end{equation}
with $\hat{\vect{h}}_k = [\hat{\vect{h}}_{1k}^{\Ttran} \, \cdots \, \hat{\vect{h}}_{Lk}^{\Ttran}]^{\Ttran}$ and the covariance matrix of the effective noise (channel estimation error, quantization distortion, and thermal noise)
\begin{equation}
    \vect{Z} = \sum_{i=1}^{K} p_i\, \vect{C}_i + \vect{C}_{\vect{e}} + \sigma_1^2 \vect{I}_M,
\end{equation}
where $\vect{C}_i = \blkdiag\left(\vect{C}_{1i}, \ldots, \vect{C}_{Li}\right)$ and $\vect{C}_{\vect{e}} = \blkdiag\left(\vect{C}_{\vect{e},1}, \ldots, \vect{C}_{\vect{e},L}\right)$ from \eqref{eq:distortion_cov}, evaluated using the long-term statistics $\diag(\vect{R}_{\vect{y},l}) = \diag\big(\sum_{i=1}^{K} p_i \vect{R}_{li}\big) + \sigma_1^2 \vect{I}_N$. The SINR in \eqref{eq:SINR} is maximized by the quantization-aware MMSE combiner
\begin{equation}
    \vect{v}_k = p_k \left( \sum_{i=1}^{K} p_i\, \hat{\vect{h}}_i \hat{\vect{h}}_i^{\Htran} + \vect{Z} \right)^{-1} \hat{\vect{h}}_k.
    \label{eq:mmse_combiner}
\end{equation}
Accounting for the access bandwidth and the fraction of the frame allocated to the access phase, the effective throughput of UE $k$ is
\begin{equation}
    R_k = t_1 B_1\, \mathsf{SE}_k \quad \text{[bit/s]}.
    \label{eq:throughput}
\end{equation}

\begin{remark}[End-to-end Bussgang-based evaluation]
\label{rem:bussgang}
The AQNM in \eqref{eq:aqnm} treats the quantization distortion as uncorrelated with the input and across antennas and time instants, which is an approximation that neglects the true distortion correlations of the nonlinear quantizer. It is adopted solely to obtain a tractable framework for the joint optimization of the resource-allocation variables and the per-AP resolutions $\{b_l\}$. In the performance evaluation, the actual Lloyd--Max quantizers are applied to both pilot and data samples, and the achievable rates are computed via the end-to-end Bussgang decomposition of the cascaded system \cite{demir2020bussgang}, using the worst-case uncorrelated additive distortion bound. This ensures that the reported rates account for the quantization-error correlations that the AQNM ignores, and that the optimized resolutions $\{b_l\}$ are validated under the exact quantizer behavior. The complete evaluation methodology is presented in Section~\ref{sec:bussgang}.
\end{remark}

\subsection{Wireless Fronthaul Phase}
\label{subsec:fronthaul}
During the fronthaul phase of duration $t_2$, the APs deliver their buffered bit streams to the CPU over the bandwidth $B_2$ (the full band being available in the TD mode, and the fronthaul subband in the FD mode). The CPU is equipped with $M_c \geq L$ antennas, and all active APs transmit \emph{simultaneously} via space-division multiple access (SDMA), being separated at the CPU by zero-forcing (ZF) combining \cite{topal2026unlocking}. Since the AP and CPU deployments are static with line-of-sight (LOS) connectivity, the fronthaul channels are assumed to be perfectly known.

Each AP transmits a single stream toward the CPU using transmit beamforming over its $N$ antennas, and the CPU separates the $L_{\mathcal{A}}$ simultaneously transmitting active APs via ZF combining over the effective channel matrix $\bar{\vect{G}}$, whose rows are the conjugated effective fronthaul channels $\vect{g}_l^{\Htran}$, $l \in \mathcal{A}$, after transmit beamforming. The achievable fronthaul rate of AP $l \in \mathcal{A}$ is \cite{topal2026unlocking}
\begin{equation}
    R_l^{\rm fh} = t_2 B_2 \log_2\left( 1 + \frac{\bar{p}_l}{N_0 B_2 D_l} \right), \quad D_l = \left[ \left( \bar{\vect{G}}^{\Htran} \bar{\vect{G}} \right)^{-1} \right]_{ll},
    \label{eq:fronthaul_rate}
\end{equation}
in bit/s, where $\bar{p}_l \in (0, P_{\rm fh}^{\max}]$ is the fronthaul transmit power of AP $l$---an optimization variable in this work---and the noise power $N_0 B_2$ scales with the fronthaul bandwidth. Note that the gains $\{D_l\}$ depend on the active set: removing an AP from $\mathcal{A}$ relaxes the ZF constraints and thus reduces (improves) the effective inverse gains $D_l$ of the remaining APs.

\emph{Fronthaul rate requirement:} The bits generated during the access phase must be fully delivered within the fronthaul phase, i.e.,
\begin{equation}
     R_l^{\rm fh} \geq F_l, \qquad l \in \mathcal{A},
    \label{eq:fronthaul_constraint}
\end{equation}
with $F_l$ given in \eqref{eq:fronthaul_load}. For any given bandwidths, time split, and resolutions, the minimum fronthaul powers that satisfy \eqref{eq:fronthaul_constraint} with equality follow in closed form by channel inversion:
\begin{equation}
    \bar{p}_l = N_0 B_2 D_l \left( 2^{\frac{F_l}{t_2 B_2}} - 1 \right), \qquad l \in \mathcal{A},
    \label{eq:zf_power}
\end{equation}
which must respect the per-AP fronthaul power budget $\bar{p}_l \leq P_{\rm fh}^{\max}$. Equation \eqref{eq:zf_power} exposes the core coupling of the integrated access--fronthaul design: increasing $B_1$, $t_1$, or $b_l$ improves the access-side rate and reduces the quantization distortion, but inflates the fronthaul load exponentially in the exponent $F_l/(t_2 B_2)$, thereby increasing the fronthaul transmit power or forcing a larger $t_2$ or $B_2$.

\subsection{Power Consumption Model}
\label{subsec:power_model}
We model the network power consumption per frame by combining the AP transceiver model of \cite{enqvist2024fundamentals,lopezperez2021energy} with the end-to-end cell-free network model of \cite{topal2026unlocking}. The total consumed energy per unit-duration frame (equivalently, the average power) is
\begin{align}
    P_{\rm tot} &= t_1 P^{\rm ac} + \left( 1 - t_1 \right) P^{\rm sl}_{\rm ac} + t_2 P^{\rm fh} + \left( 1 - t_2 \right) P^{\rm sl}_{\rm fh} \nonumber \\
    &\quad + \left( L - L_{\mathcal{A}} \right) P_{\rm sl} + P^{\rm fix},
    \label{eq:total_power}
\end{align}
where $P^{\rm sl}_{\rm ac} = L_{\mathcal{A}}\, \delta_{\rm s} \left( \mu_{\rm AP} + D_0 N \right)$ is the sleep power of the access transceivers of the active APs while the access link is inactive, $P^{\rm sl}_{\rm fh} = \delta_{\rm s} \big(  L_{\mathcal{A}} P_{0,{\rm fh}} + \mu_{\rm CPU}^{\rm fh} + D_0^{\rm c} M_c \big)$ is the corresponding sleep power of the fronthaul radios, and each switched-off AP consumes the constant sleep power $P_{\rm sl} = \delta_{\rm s} \left( \mu_{\rm AP} + D_0 N + P_{0,{\rm fh}} \right)$ during the entire frame. Here, $\delta_{\rm s} \in [0, 1]$ is the sleep-depth parameter determined by the hardware capabilities \cite{enqvist2026sleep}, and active operation is assumed more power-hungry than sleeping, i.e., $P^{\rm ac} > P^{\rm sl}_{\rm ac}$ and $P^{\rm fh} > P^{\rm sl}_{\rm fh}$. In the TD mode, $1 - t_1 = t_2$ and $1 - t_2 = t_1$, so \eqref{eq:total_power} captures that each radio interface sleeps while the other link occupies the band; in the FD mode, the idle fractions correspond to genuine sleep intervals. The remaining constituents are detailed below.

\subsubsection{Access-phase power} During $t_1$, the active APs receive over the bandwidth $B_1$ and the UEs transmit:
\begin{equation}
    P^{\rm ac} = \sum_{l \in \mathcal{A}} \left( \mu_{\rm AP} + \left( D_0 + \nu B_1 \right) N \right) + \sum_{k=1}^{K} \left( \frac{p_k}{\kappa_{\rm UE}} + P_{0,{\rm UE}} \right),
    \label{eq:access_power}
\end{equation}
where $\mu_{\rm AP}$ is the fixed AP circuit power (local oscillator, control), $D_0$ is the static power per receive chain, $\nu B_1$ captures the bandwidth-proportional per-antenna processing power (analog-to-digital conversion and baseband operations running at the sampling rate), with a typical value of $\nu = 10^{-10}$\,J/sample \cite{enqvist2024fundamentals}, $\kappa_{\rm UE} \in (0,1]$ is the UE power amplifier efficiency, and $P_{0,{\rm UE}}$ is the UE circuit power.

\subsubsection{Fronthaul-phase power} During $t_2$, all active APs simultaneously transmit their bit streams and the CPU radio receives:
\begin{equation}
    P^{\rm fh} = \sum_{l \in \mathcal{A}} \left( \frac{\bar{p}_l}{\kappa_{\rm fh}} + P_{0,{\rm fh}} \right) + \mu_{\rm CPU}^{\rm fh} + \left( D_0^{\rm c} + \nu_{\rm c} B_2 \right) M_c,
    \label{eq:fronthaul_power}
\end{equation}
where $\kappa_{\rm fh}$ is the fronthaul power amplifier efficiency, $P_{0,{\rm fh}}$ is the static power of the AP fronthaul radio \cite{topal2026unlocking}, and the last two terms model the CPU-side fronthaul receiver with $M_c$ antennas, analogously to the access-phase transceiver model.

\subsubsection{Load-independent and processing power} The CPU consumes a fixed power and a rate-dependent processing/decoding power:
\begin{equation}
    P^{\rm fix} = P_{\rm CPU} + \eta_{\rm dec} \sum_{k=1}^{K} R_k,
    \label{eq:fixed_power}
\end{equation}
where $\eta_{\rm dec}$ [W per bit/s] models the decoding and centralized-processing power that scales with the delivered throughput \cite{enqvist2026sleep}; typical values are $10^{-11}$\,J/bit for the rate-proportional coding and backhaul signaling \cite{enqvist2024fundamentals} and $0.1$ and $0.8$\,W per Gbit/s for the data coding and decoding, respectively \cite{bjornson2017massive}. A finer, giga operations per second (GOPS)-based model of the centralized processing load as in \cite{topal2026unlocking} can be substituted without changing the structure of the optimization problem.

\subsection{Energy Efficiency Maximization Problem}
\label{subsec:problem}
The network energy efficiency (EE) is defined as the total number of delivered information bits per unit of consumed energy,
\begin{equation}
    \mathsf{EE}\left(B_1, B_2, t_1, t_2, \{b_l\}, \{\bar{p}_l\}\right) = \frac{\sum_{k=1}^{K} R_k}{P_{\rm tot}} \quad \text{[bit/Joule]},
    \label{eq:EE}
\end{equation}
where the throughput $R_k$ in \eqref{eq:throughput} depends on $(t_1, B_1, \{b_l\})$ through the SE in \eqref{eq:SE}--\eqref{eq:SINR} and the distortion factors $\eta_l \approx c_q 2^{-2 b_l}$, while $P_{\rm tot}$ in \eqref{eq:total_power} depends on all optimization variables, including the fronthaul powers $\{\bar{p}_l\}$, whose minimum feasible values are given in closed form in \eqref{eq:zf_power}.

To treat the two operating modes of Section~\ref{subsec:frame} in a unified manner, we define the mode-dependent time--bandwidth feasibility sets
\begin{align}
    \mathcal{X}_{\rm TD} &= \big\{ (B_1, B_2, t_1, t_2) : t_1 + t_2 = 1, \; t_1, t_2 > 0, \nonumber \\
    &\qquad\quad 0 < B_1 \leq B, \; 0 < B_2 \leq B \big\},
    \label{eq:XTD} \\
    \mathcal{X}_{\rm FD} &= \big\{ (B_1, B_2, t_1, t_2) : 0 < t_1 \leq 1, \; 0 < t_2 \leq 1, \nonumber \\
    &\qquad\quad B_1, B_2 > 0, \; B_1 + B_2 \leq B \big\},
    \label{eq:XFD}
\end{align}
and let $\mathcal{X} \in \{ \mathcal{X}_{\rm TD}, \mathcal{X}_{\rm FD} \}$ denote the feasibility set of the selected mode. The joint access--fronthaul--sleep resource allocation and adaptive quantization resolution problem is formulated as
\begin{subequations}
\label{eq:problem}
\begin{align}
    \underset{\substack{B_1, B_2, t_1, t_2, \\ \{b_l\}, \{\bar{p}_l\}}}{\mathrm{maximize}} \quad & \mathsf{EE}\left(B_1, B_2, t_1, t_2, \{b_l\}, \{\bar{p}_l\}\right)
    \label{eq:objective} \\
    \mathrm{subject~to} \quad
    &  R_l^{\rm fh} \geq 2 N B_1 t_1 \left( \tfrac{\tau_p}{\tau_c} b_{\rm p} + \tfrac{\tau_u}{\tau_c} b_l \right), \; \forall l \in \mathcal{A},
    \label{eq:c_fronthaul} \\
    & 0 < \bar{p}_l \leq P_{\rm fh}^{\max}, \quad \forall l \in \mathcal{A},
    \label{eq:c_fhpower} \\
    & \left( B_1, B_2, t_1, t_2 \right) \in \mathcal{X},
    \label{eq:c_mode} \\
    & b_l \in \{0, 1, \ldots, b_{\max}\}, \quad l = 1, \ldots, L,
    \label{eq:c_bits} \\
    & R_k \geq R_{\min}, \quad k = 1, \ldots, K,
    \label{eq:c_qos}
\end{align}
\end{subequations}
where \eqref{eq:c_fronthaul} guarantees that every AP empties its buffer within the fronthaul phase, \eqref{eq:c_fhpower} enforces the per-AP fronthaul power budget, \eqref{eq:c_mode} enforces the time--bandwidth feasibility region of the selected operating mode, \eqref{eq:c_bits} restricts the per-AP quantization resolutions to practical integer values---with $b_l = 0$ corresponding to putting AP $l$ to sleep for the entire frame---and the optional quality-of-service constraint \eqref{eq:c_qos} imposes a minimum throughput per UE.

Problem \eqref{eq:problem} is a mixed-integer, nonconvex fractional program: the objective is a ratio of nonconvex functions, the variables $(B_1, t_1,\{b_l\})$ and $(B_2, t_2, \{\bar{p}_l\})$ are coupled through the fronthaul constraint \eqref{eq:c_fronthaul}, the resolutions $\{b_l\}$ are integers that enter both the rate (through $\eta_l$) and the fronthaul load. In Section~\ref{sec:algorithm}, we develop an alternating optimization algorithm in which every block update admits either a closed-form solution or an efficient one-dimensional search; the resulting design is then validated through the end-to-end Bussgang-based rate evaluation of Remark~\ref{rem:bussgang}.

\begin{remark}[Fixed uplink UE transmit powers]
\label{rem:uepower}
Throughout this work, the uplink UE transmit powers are fixed at their maximum value, i.e., $p_k = p_{\max}$ for all $k$, and are therefore not included as optimization variables in \eqref{eq:problem}. This simplification is motivated by the primary focus of the paper on network-side resource allocation, namely the optimization of the time split, bandwidth allocation, quantization resolutions, and fronthaul transmit powers. These network-side components dominate the total power consumption in \eqref{eq:total_power}, whereas the UE power contribution,
$
t_1 \sum_{k=1}^{K} \left( \frac{p_k}{\kappa_{\rm UE}} + P_{0,{\rm UE}} \right),$
represents only a comparatively small fraction of the overall power budget.
\end{remark}

\vspace{-4mm}

\section{Alternating Optimization for EE Maximization}
\label{sec:algorithm}

Problem \eqref{eq:problem} couples continuous resource-allocation variables, integer quantization resolutions, and per-AP fronthaul powers through the nonconvex fractional objective \eqref{eq:EE} and the fronthaul delivery constraint \eqref{eq:c_fronthaul}, so a joint search is intractable even for moderate $L$. Instead, we exploit the fact that the problem exhibits a favorable structure \emph{per block of variables} and propose an alternating optimization procedure over four blocks: (i) the time split $(t_1, t_2)$; (ii) the access bandwidth---updated alone in the TD mode, and \emph{jointly with the fronthaul bandwidth} in the FD mode, where the band budget $B_1 + B_2 \leq B$ couples the two bandwidths and $B_2$ is set to the minimum fronthaul bandwidth induced by $B_1$; (iii) the fronthaul bandwidth and quantization resolutions $(B_2, \{b_l\})$; and (iv) the fronthaul bandwidth--power pair $(B_2, \{\bar{p}_l\})$. Every block update admits either a closed-form solution or a scalar bisection/grid search, and no update decreases the EE, which guarantees monotone convergence. The complete procedure is summarized in Algorithm~\ref{alg:ao} and applies to both the TD and the FD modes. Steps 1, 3, and 4 share the same structure in the two modes and differ only through the mode-dependent time and bandwidth bounds highlighted in each step; the sole structural difference is Step 2, which optimizes $B_1$ alone in the TD mode (for a fixed $B_2$), whereas in the FD mode it updates the pair $(B_1, B_2)$ jointly along the coupling $B_2 = B_2^{\min}(B_1)$ dictated by the most binding fronthaul delivery constraint. In addition, the sweep in Step 3 embeds a sleep-mode-based AP switch-off mechanism: whenever the fronthaul cannot support even the minimum resolution $b_l = 1$ for an AP, that AP is put to sleep ($b_l = 0$), the fronthaul ZF gains and the receive combiners are re-derived over the remaining active APs, and the optimization proceeds with the reduced active set.

Throughout this section, we write the fronthaul load in \eqref{eq:fronthaul_load} as $F_l = t_1 \tilde{F}_l$, where
\begin{equation}
    \tilde{F}_l = 2 N B_1 \left( \tfrac{\tau_p}{\tau_c} b_{\rm p} + \tfrac{\tau_u}{\tau_c} b_l \right)
    \label{eq:bitrate}
\end{equation}
is the fronthaul bit-arrival rate of AP $l$ during the access phase, i.e., the number of bits per second that AP $l$ generates while the access link is active. Moreover, we define 
\begin{equation}
    \tilde{R}_l^{\rm fh} =\frac{R_l^{\rm fh}}{t_2}= B_2 \log_2\left( 1 + \frac{\bar{p}_l}{N_0 B_2 D_l} \right).
    \label{eq:fronthaul_rate2}
\end{equation}

\subsection{Step 1: Optimal Time Split $(t_1, t_2)$}
\label{subsec:step1}

We first optimize the time split for fixed $(B_1, B_2, \{b_l\}, \{\bar{p}_l\})$. The key observation is that both the throughput and the consumed energy are linear in $(t_1, t_2)$, so the optimal split admits a closed form.

\begin{proposition}[Optimal time split]
\label{prop:tsplit}
For fixed $(B_1, B_2, \{b_l\}, \{\bar{p}_l\})$, define the \emph{fronthaul expansion factor}
\begin{equation}
    \omega = \max_{l \in \mathcal{A}} \frac{\tilde{F}_l}{\tilde{R}_l^{\rm fh}},
    \label{eq:omega}
\end{equation}
i.e., the minimum fronthaul transmission time required per unit of access time, dictated by the most binding AP. Then, the optimal solution of \eqref{eq:problem} restricted to $(t_1, t_2)$ is
\begin{align}
    &\text{TD mode:} \quad t_1^\star = \frac{1}{1 + \omega}, \qquad\; t_2^\star = \frac{\omega}{1 + \omega},
    \label{eq:tsplit} \\
    &\text{FD mode:} \quad t_1^\star = \min\left\{ 1, \tfrac{1}{\omega} \right\}, \quad t_2^\star = \omega\, t_1^\star.
    \label{eq:tsplit_fd}
\end{align}
\end{proposition}
\begin{IEEEproof}
In both modes, the fronthaul constraint \eqref{eq:c_fronthaul} is equivalent to $t_2 \geq \omega t_1$. Consider first the TD mode and substitute $t_2 = 1 - t_1$ from $\mathcal{X}_{\rm TD}$: the denominator of the EE in \eqref{eq:EE} becomes the affine function $t_1 \left( P^{\rm ac} + P^{\rm sl}_{\rm fh} - P^{\rm sl}_{\rm ac} - P^{\rm fh} \right) + \beta_{\rm TD}$ with the constant part $\beta_{\rm TD} = P^{\rm sl}_{\rm ac} + P^{\rm fh} + (L - L_{\mathcal{A}}) P_{\rm sl} + P^{\rm fix} > 0$, while the numerator is $t_1 B_1 \sum_{k} \mathsf{SE}_k$. The derivative of such a linear-fractional function has the sign of $\beta_{\rm TD} > 0$, so the EE is strictly increasing in $t_1$, and the optimum saturates the feasibility bound $1 - t_1 \geq \omega t_1$, yielding \eqref{eq:tsplit}. In the FD mode, for fixed $t_1$, the EE is strictly decreasing in $t_2$, because the numerator does not depend on $t_2$ while the denominator grows at the rate $P^{\rm fh} - P^{\rm sl}_{\rm fh} > 0$; hence, $t_2 = \omega t_1$ at the optimum. Substituting, the denominator becomes $t_1 \left[ \left( P^{\rm ac} - P^{\rm sl}_{\rm ac} \right) + \omega \left( P^{\rm fh} - P^{\rm sl}_{\rm fh} \right) \right] + \beta_{\rm FD}$ with $\beta_{\rm FD} = P^{\rm sl}_{\rm ac} + P^{\rm sl}_{\rm fh} + (L - L_{\mathcal{A}}) P_{\rm sl} + P^{\rm fix} > 0$, so the EE is again strictly increasing in $t_1$, and the optimum is limited only by $t_1 \leq 1$ and $t_2 = \omega t_1 \leq 1$ from $\mathcal{X}_{\rm FD}$, yielding \eqref{eq:tsplit_fd}.
\end{IEEEproof}

Proposition~\ref{prop:tsplit} admits an intuitive reading. In the TD mode, idle time is never beneficial because the sleep and fixed powers are consumed regardless of the split, so the frame is packed with as much access time as the fronthaul can absorb, and $\omega$ quantifies how much fronthaul time each second of access time costs. In the FD mode, the same reasoning pushes the access link to be active during the whole frame ($t_1^\star = 1$) whenever the fronthaul subband is fast enough ($\omega \leq 1$), in which case the fronthaul radio sleeps for the fraction $1 - \omega$ of the frame; if instead $\omega > 1$, the fronthaul becomes the bottleneck and the access activity must be throttled to $t_1^\star = 1/\omega < 1$.

\subsection{Step 2: Access-Bandwidth Update (TD: $B_1$; FD: Joint $(B_1, B_2)$)}
\label{subsec:step2}

We next update the access bandwidth for fixed $(t_1, t_2, \{b_l\},\{\bar{p}_l\})$. In the TD mode, $B_1$ is optimized for a fixed $B_2$; in the FD mode, where the two bandwidths share the band budget, $B_1$ and $B_2$ are updated \emph{jointly}, as described at the end of this subsection. Consider first the TD mode, in which the search interval is simply $0 < B_1 \leq B_1^{\max}$ with $B_1^{\max} = B$: the fronthaul delivery constraints are not imposed within this step---they are restored immediately afterwards by Steps 3--4, which re-decide the resolutions and the fronthaul powers at the updated bandwidth---and the safeguard rule of Section~\ref{subsec:overall} ensures that only EE-improving updates are retained.

The access bandwidth affects the SE in two distinct ways. Besides the \emph{direct} dependence through the noise power $\sigma_1^2 = N_0 B_1$ in the data-domain effective-noise covariance, there is an \emph{indirect} dependence through the channel estimation quality: the matrix $\bm{\Psi}_{lk}$ in the LMMSE estimator contains $\sigma_1^2$ (and the pilot-quantization distortion), so the channel estimates $\{\hat{\vect{h}}_{lk}\}$, the error covariances $\{\vect{C}_{lk}\}$, and consequently all expected values entering the SE expression \eqref{eq:SE}--\eqref{eq:SINR} vary with $B_1$. To obtain a tractable update, we deliberately keep this indirect dependence \emph{frozen} during the $B_1$ optimization: the receive combiners, the channel estimates, and the estimation-error statistics are fixed at the values computed with the $B_1$ of the current iterate, and only the direct noise scaling is retained. This simplification is adopted purely for tractability and is compensated at every outer iteration: once the block updates of the iteration are completed, the estimates, the error statistics, and the combiners are re-computed with the updated $B_1$, so that the frozen quantities always track the current operating point. Together with the safeguarded acceptance rule of Section~\ref{subsec:overall}, this ensures that the simplification never decreases the exact EE along the iterations. Under this convention, recall from \eqref{eq:distortion_cov} that the effective-noise covariance decomposes as
\begin{equation}
    \vect{Z} = \underbrace{\sum_{i=1}^{K} p_i \vect{C}_i + \vect{C}_{\vect{e}}^{(0)}}_{\displaystyle \triangleq\, \vect{Z}^{(0)}} + \; N_0 B_1 \left( \vect{I}_M + \bm{\Lambda} \right),
    \label{eq:Zdecomp}
\end{equation}
where $\vect{C}_{\vect{e}}^{(0)}$ collects the signal-driven part of the quantization distortion, obtained by replacing $\diag(\vect{R}_{\vect{y},l})$ in \eqref{eq:distortion_cov} with $\diag\big( \sum_i p_i \vect{R}_{li} \big)$, and $\bm{\Lambda} = \blkdiag\big( \tfrac{\eta_1}{1-\eta_1} \vect{I}_N, \ldots, \tfrac{\eta_L}{1-\eta_L} \vect{I}_N \big)$ scales the noise-driven part: the thermal noise enters the quantizer input, so its distortion image also grows linearly with $N_0 B_1$. Consequently, for each channel realization, the effective SINR of UE $k$ in \eqref{eq:SINR} with a fixed combiner takes the form
\begin{equation}
    \mathsf{SINR}_k(B_1) = \frac{a_k}{d_k B_1 + c_k},
    \label{eq:sinr_structure}
\end{equation}
with the constants
\begin{align}
    a_k &= p_k \big| \vect{v}_k^{\Htran} \hat{\vect{h}}_k \big|^2, \qquad
    d_k = N_0\, \vect{v}_k^{\Htran} \left( \vect{I}_M + \bm{\Lambda} \right) \vect{v}_k, \nonumber \\
    c_k &= \sum_{i \neq k} p_i \big| \vect{v}_k^{\Htran} \hat{\vect{h}}_i \big|^2 + \vect{v}_k^{\Htran} \vect{Z}^{(0)} \vect{v}_k,
    \label{eq:acd}
\end{align}
all of which are positive and independent of $B_1$. Approximating the expectation in \eqref{eq:SE} by a sample average over the channel realizations of the current iterate, the sum throughput becomes
\begin{equation}
    g(B_1) = t_1 \frac{\tau_u}{\tau_c} \frac{1}{T}\sum_{m} B_1 \log_2\left( 1 + \frac{a_m}{d_m B_1 + c_m} \right),
    \label{eq:gB1}
\end{equation}
where the index $m$ runs over all (UE, realization) pairs and $T$ is the number of sample realizations. On the power side, only the bandwidth-proportional AP processing power and the decoding power depend on $B_1$:
\begin{equation}
    P_{\rm tot}(B_1) = \bar{k} B_1 + \bar{l} + \eta_{\rm dec}\, g(B_1),
    \label{eq:PtotB1}
\end{equation}
with $\bar{k} = t_1 \nu N L_{\mathcal{A}} > 0$ and $\bar{l} > 0$ collecting all $B_1$-independent terms. Therefore,
\begin{equation}
    \mathsf{EE}(B_1) = \frac{g(B_1)}{\bar{k} B_1 + \bar{l} + \eta_{\rm dec}\, g(B_1)} = \frac{1}{f(B_1) + \eta_{\rm dec}},
    \label{eq:EEB1}
\end{equation}
with $f(B_1) = \frac{\bar{k} B_1 + \bar{l}}{g(B_1)}$, so maximizing the EE over $B_1$ is equivalent to minimizing $f(B_1)$ over $(0, B]$. The following lemma shows that a scalar bisection suffices.

\begin{lemma}
\label{lem:bisection}
Let $g(B) = \sum_{m} w_m B \log_2\left(1 + \frac{a_m}{d_m B + c_m}\right)$ with $w_m, a_m, c_m, d_m > 0$, and let $f(B) = \frac{\bar{k} B + \bar{l}}{g(B)}$ with $\bar{k} \geq 0$ and $\bar{l} > 0$. Then, $g$ is strictly increasing and strictly concave on $(0, \infty)$, and
\begin{equation}
    \phi(B) = \bar{k}\, g(B) - \left( \bar{k} B + \bar{l} \right) g'(B)
    \label{eq:phi}
\end{equation}
is strictly increasing with $\phi(0^+) < 0$. Consequently, $f$ has a unique global minimizer on $(0, B]$, given by $\min\{B_\phi, B\}$, where $B_\phi$ is the unique root of $\phi(B) = 0$, obtained by bisection; if $\phi(B) \leq 0$, the minimizer is $B$.
\end{lemma}
\begin{IEEEproof}
Each summand of $g$ can be written as $\frac{w_m}{\ln (2)} B \ln\left(1 + \frac{A}{B + C}\right)$ with $A = a_m / d_m > 0$ and $C = c_m / d_m > 0$. The first derivative of $B \ln\left(1 + \frac{A}{B + C}\right)$ is $\ln\left(1 + \frac{A}{B+C}\right) - \frac{AB}{(B+C)(A+B+C)}$ and is positive for all $B > 0$: the inequality $\ln(1+x) > \frac{x}{1+x}$ with $x = \frac{A}{B+C}$ gives the lower bound $\frac{A}{A+B+C}$, which dominates $\frac{AB}{(B+C)(A+B+C)}$ since $B + C > B$. Its second derivative equals
\begin{equation}
    -\frac{A \left( A(B+C) + 2C(B+C) + AC \right)}{(B+C)^2 (B+C+A)^2} < 0.
\end{equation}
Hence, $g$ is strictly increasing and strictly concave. It follows that $\phi'(B) = -(\bar{k} B + \bar{l})\, g''(B) > 0$, so $\phi$ is strictly increasing. Moreover, $g(0^+) = 0$ and $g'(0^+) = \sum_m w_m \log_2(1 + a_m / c_m) \in (0, \infty)$ yield $\phi(0^+) = -\bar{l}\, g'(0^+) < 0$. Since $f'(B) = \phi(B)/g^2(B)$ has the same sign as $\phi(B)$, the function $f$ is strictly decreasing before the unique root of $\phi$ and strictly increasing after it, which proves the claim.
\end{IEEEproof}

\emph{Joint $(B_1, B_2)$ update in the FD mode:} In the FD mode, updating $B_1$ at a fixed $B_2$ would confine the search to $B_1 \leq B - B_2$, which is unnecessarily restrictive. Instead, we exploit the fact that, at the optimum of the joint update, the most binding delivery constraint holds with equality at the fronthaul powers $\{\bar{p}_l\}$ inherited from the previous iteration: for a given $B_1$, the minimum fronthaul bandwidth that AP $l$ requires is the unique root of
\begin{equation}
    t_2\, B_2 \log_2\left( 1 + \frac{\bar{a}_l^{\rm fh}}{B_2} \right) = t_1 \tilde{F}_l, \qquad \bar{a}_l^{\rm fh} = \frac{\bar{p}_l}{N_0 D_l},
    \label{eq:B2min}
\end{equation}
since the left-hand side is strictly increasing and concave in $B_2$; taking the maximum over $l \in \mathcal{A}$ yields $B_2^{\min}(B_1)$, which is strictly increasing in $B_1$. Substituting this coupling into the objective makes the EE one-dimensional in $B_1$: the throughput retains the concave structure of \eqref{eq:gB1}, while the denominator now depends on $B_1$ both through the affine access-side term in \eqref{eq:PtotB1} and through the bandwidth-proportional CPU-side fronthaul power $t_2\nu_{\rm c} M_c B_2^{\min}(B_1)$; hence, the exact ratio structure of Lemma~\ref{lem:bisection} no longer applies verbatim. Instead, the band budget $B_1 + B_2^{\min}(B_1) \leq B$ first determines the maximal feasible $B_1$ by bisection (the left-hand side being strictly increasing in $B_1$), and the resulting one-dimensional EE along the curve $B_1 \mapsto \big( B_1, B_2^{\min}(B_1) \big)$ is then maximized by a second bisection on the sign of its derivative, exploiting its unimodality. 

\subsection{Step 3: Fronthaul Bandwidth and Resolutions $(B_2, \{b_l\})$}
\label{subsec:step3}

We now update the fronthaul bandwidth jointly with the integer resolutions for fixed $(t_1, t_2, B_1,\{\bar{p}_l\})$. A brute-force search over $\{1, \ldots, b_{\max}\}^L$ is exponential in $L$; the following monotonicity property removes this bottleneck.

\begin{proposition}
\label{prop:saturating}
Fix $(t_1, t_2, B_1, B_2)$ and the fronthaul powers $\{\bar{p}_l\}$. Then, the EE is nondecreasing in every resolution $b_l$ over the feasible set of \eqref{eq:c_fronthaul}. Consequently, the optimal resolution of each AP is the largest fronthaul-feasible one.
\end{proposition}
\begin{IEEEproof}
Increasing $b_l$ strictly decreases $\eta_l$ and hence the loading factor $\tfrac{\eta_l}{1 - \eta_l}$, so the distortion covariance $\vect{C}_{\vect{e}}$ decreases in the positive-semidefinite (Loewner) order, while the channel estimates are unaffected because the pilot resolution $b_{\rm p}$ is fixed. Writing the SINR under MMSE combining as $\mathsf{SINR}_k = p_k \hat{\vect{h}}_k^{\Htran} \big( \sum_{i \neq k} p_i \hat{\vect{h}}_i \hat{\vect{h}}_i^{\Htran} + \vect{Z} \big)^{-1} \hat{\vect{h}}_k$ and noting that $\vect{A} \succeq \vect{B} \succ \vect{0}$ implies $\vect{A}^{-1} \preceq \vect{B}^{-1}$, every $\mathsf{SINR}_k$ is nondecreasing in $b_l$. The throughput therefore does not decrease, while the denominator $P_{\rm tot}$ is unchanged, since neither the access-side power nor the (fixed) fronthaul powers depend on $\{b_l\}$.
\end{IEEEproof}

Motivated by Proposition~\ref{prop:saturating}, we evaluate the feasibility of each resolution at the fronthaul powers $\{\bar{p}_l\}$ inherited from the previous iteration---consistent with Step 2, the powers are kept fixed while the resolutions are being decided---and define the \emph{constraint-saturating} resolution
\begin{equation}
    b_l^\star(B_2) = \max\big\{ b \leq b_{\max} : t_1 \tilde{F}_l(b) \leq t_2\, \tilde{R}_l^{\rm fh}(B_2) \big\},
    \label{eq:bstar}
\end{equation}
where $\tilde{F}_l(b)$ denotes \eqref{eq:bitrate} evaluated at resolution $b$ and $\tilde{R}_l^{\rm fh}(B_2) = B_2 \log_2\big(1 + \frac{\bar{p}_l}{N_0 B_2 D_l}\big)$ is the fronthaul rate of AP $l$ at the inherited power; the minimum equality powers subsequently assigned by the sweep can then never exceed $\{\bar{p}_l\}$, so feasibility is preserved by construction. Since $\tilde{F}_l(b)$ is affine in $b$, \eqref{eq:bstar} is solved in closed form followed by flooring.

The residual dependence on $B_2$ is one-dimensional but nonconvex, because $b_l^\star(B_2)$ is integer-valued and piecewise constant. We therefore sweep $B_2$ over a finite grid $\mathcal{G} \subset (0, B_2^{\rm ub}]$---where $B_2^{\rm ub} = B$ in the TD mode and $B_2^{\rm ub} = B - B_1$ in the FD mode---and embed the AP switch-off mechanism into the sweep. For each candidate $B_2' \in \mathcal{G}$: (i) the resolutions of the currently active APs are set to $b_l^\star(B_2')$; (ii) if \eqref{eq:bstar} admits no $b \geq 1$ for some AP, that AP is put to sleep in the candidate, i.e., $b_l = 0$, yielding the candidate active set $\mathcal{A}' \subseteq \mathcal{A}$; (iii) the ZF gains $\{D_l\}$ are re-computed over $\mathcal{A}'$---they can only improve, which may in turn increase the resolutions of the surviving APs, so (i)--(iii) are repeated until the active set stabilizes; (iv) the receive combiners are re-derived from the channel estimates of the active APs only, by zeroing the blocks of the sleeping APs; and (v) the EE of the candidate tuple, including the sleep power of the $L - |\mathcal{A}'|$ sleeping APs, is evaluated. The candidate with the largest EE is retained; the current iterate is always included among the candidates, so the step never decreases the EE. If the selected candidate involves switch-offs, they are \emph{irreversible}: once an AP is put to sleep, it remains asleep in all subsequent outer iterations, which prevents on--off oscillations and repeated reconfiguration signaling. 

\subsection{Step 4: Fronthaul Bandwidth--Power Refinement $(B_2, \{\bar{p}_l\})$}
\label{subsec:step4}

With $\{b_l\}$ (and thus the access-side throughput) fixed, the EE is maximized by minimizing $P_{\rm tot}$, of which only the fronthaul-phase part depends on $(B_2, \{\bar{p}_l\})$. At any optimal point, the fronthaul constraints \eqref{eq:c_fronthaul} hold with equality---otherwise some $\bar{p}_l$ could be reduced without affecting the throughput---so the powers obey the channel-inversion rule \eqref{eq:zf_power},
\begin{equation}
    \bar{p}_l(B_2) = N_0 B_2 D_l \left( 2^{\bar{a}_l / B_2} - 1 \right), \qquad \bar{a}_l = \frac{F_l}{t_2},
    \label{eq:pbar_B2}
\end{equation}
and the refinement reduces to a scalar problem in $B_2$. Substituting \eqref{eq:pbar_B2} into $t_2 P^{\rm fh}$, the $B_2$-dependent part of the consumed energy is
\begin{equation}
    h(B_2) = t_2 \sum_{l \in \mathcal{A}} \frac{N_0 D_l}{\kappa_{\rm fh}}\, B_2 \left( 2^{\bar{a}_l / B_2} - 1 \right) + t_2\, \nu_{\rm c} M_c B_2.
    \label{eq:hB2}
\end{equation}

\begin{lemma}
\label{lem:B2convex}
The function $h$ in \eqref{eq:hB2} is strictly convex on $(0, \infty)$ with a unique minimizer, and each $\bar{p}_l(B_2)$ in \eqref{eq:pbar_B2} is strictly decreasing in $B_2$.
\end{lemma}
\begin{IEEEproof}
The map $B_2 \mapsto B_2 \big( 2^{\bar{a}_l / B_2} - 1 \big)$ is the perspective of the convex function $u \mapsto 2^{u} - 1$ evaluated at $u = \bar{a}_l$, hence convex; its second derivative $\frac{\bar{a}_l^2 \ln^2 (2)}{B_2^3}\, 2^{\bar{a}_l / B_2} > 0$ shows strict convexity, and adding the linear term preserves it. Its first derivative is $2^{\bar{a}_l / B_2} \big( 1 - \frac{\bar{a}_l \ln (2)}{B_2} \big) - 1 < 0$ for all $B_2 > 0$, by the inequality $e^{y} (1 - y) < 1$ for $y = \frac{\bar{a}_l \ln(2)}{B_2} > 0$; multiplying by $N_0 D_l > 0$ shows that $\bar{p}_l(B_2)$ is strictly decreasing. Finally, $h'(B_2) \to -\infty$ as $B_2 \to 0^+$ and $h'(B_2) \to t_2 \nu_{\rm c} M_c > 0$ as $B_2 \to \infty$, so the strictly increasing derivative
\begin{equation}
    h'(B_2) = t_2 \sum_{l \in \mathcal{A}} \frac{N_0 D_l}{\kappa_{\rm fh}} \left[ 2^{\frac{\bar{a}_l}{B_2}} \left( 1 - \frac{\bar{a}_l \ln 2}{B_2} \right) - 1 \right] + t_2 \nu_{\rm c} M_c
    \label{eq:hprime}
\end{equation}
has a unique root, which is the unique minimizer of $h$.
\end{IEEEproof}

By Lemma~\ref{lem:B2convex}, the minimizer of $h$ is found by bisection on $h'$, and the per-AP power cap defines a lower bound $B_2^{\min}$: since each $\bar{p}_l(B_2)$ is strictly decreasing, $B_2^{\min}$ is the smallest bandwidth for which $\max_l \bar{p}_l(B_2) \leq P_{\rm fh}^{\max}$, also found by bisection. The refined bandwidth is the projection of the root of $h'$ onto $[B_2^{\min}, B_2^{\rm ub}]$, and the fronthaul powers are set to the equality values \eqref{eq:pbar_B2} at the refined $B_2$. This step keeps the throughput unchanged and never increases the consumed energy, so it never decreases the EE.

\vspace{-2mm}
\subsection{Overall Algorithm and Convergence}
\label{subsec:overall}

The complete procedure is summarized in Algorithm~\ref{alg:ao}. To guarantee monotonicity despite the frozen-statistics conventions adopted within the iterations, a \emph{safeguarded acceptance rule} is applied to \emph{every} block update (Steps 1--4): the EE of the proposed update is evaluated---under the within-iteration convention of frozen estimation statistics and, in Steps 1--2, frozen combiners---and the update is accepted only if it does not decrease the EE; otherwise, the previous value of the corresponding block is retained. In Step 3, the rule is realized by always keeping the incumbent among the sweep candidates. At the end of every outer iteration, the channel estimates, their error statistics, and the receive combiners are re-computed at the updated $B_1$, and the EE is re-evaluated accordingly.

\begin{algorithm}[t]
\caption{Alternating EE maximization for problem \eqref{eq:problem}}
\label{alg:ao}
\begin{algorithmic}[1]
\REQUIRE Operating mode (TD or FD), statistics $\{\vect{R}_{lk}\}$, fronthaul gains $\{D_l\}$, power-model parameters, grid $\mathcal{G} \subset (0, B_2^{\rm ub}]$, tolerances $\epsilon, \epsilon_B > 0$
\STATE $\mathcal{A} \leftarrow \{1, \ldots, L\}$; $(B_1, B_2) \leftarrow (B, B)$ for TD or $(B/2, B/2)$ for FD; $\bar{p}_l \leftarrow P_{\rm fh}^{\max}$, $\forall l$; compute estimates and combiners; $b_l \leftarrow b_l^\star(B_2)$ via \eqref{eq:bstar}, $\forall l$; $n \leftarrow 0$; $\mathsf{EE}^{(0)} \leftarrow 0$
\REPEAT
    \STATE $n \leftarrow n + 1$
    \STATE \textbf{(Step 1)} $\omega \leftarrow \max_{l \in \mathcal{A}} \tilde{F}_l / \tilde{R}_l^{\rm fh}$; \ $(t_1, t_2) \leftarrow$ \eqref{eq:tsplit} for TD, or \eqref{eq:tsplit_fd} for FD
    \STATE \textbf{(Step 2)} TD: compute the coefficients in \eqref{eq:acd}; $B_1 \leftarrow \min\{ B_\phi, B \}$, where $B_\phi$ is the root of $\phi$ in \eqref{eq:phi}, found by bisection with tolerance $\epsilon_B$.\\ FD: jointly update $(B_1, B_2)$ along $B_2^{\min}(B_1)$ in \eqref{eq:B2min} by bisection, subject to $B_1 + B_2^{\min}(B_1) \leq B$, with the fronthaul powers kept fixed
    \STATE \textbf{(Step 3)} sweep the fronthaul bandwidth over the candidate grid:
    \FORALL{$B_2' \in \mathcal{G} \cup \{ B_2 \}$}
        \STATE $b_l' \leftarrow b_l^\star(B_2')$ via \eqref{eq:bstar}, $\forall l \in \mathcal{A}$; put APs with no feasible $b \geq 1$ to sleep ($b_l' \leftarrow 0$), yielding $\mathcal{A}' \subseteq \mathcal{A}$; re-compute $\{D_l\}$ over $\mathcal{A}'$ and repeat until $\mathcal{A}'$ stabilizes; re-derive combiners over $\mathcal{A}'$
        \STATE  Evaluate the EE of $\big(B_2', \{b_l'\}, \mathcal{A}'\big)$, including the sleep power
    \ENDFOR
    \STATE $(B_2, \{b_l\}, \mathcal{A}) \leftarrow$ the candidate with the largest EE (selected switch-offs are irreversible)
    \STATE \textbf{(Step 4)} $B_2 \leftarrow \Pi_{[B_2^{\min},\, B_2^{\rm ub}]}\big( \text{root of } h' \text{ in } \eqref{eq:hprime} \big)$; \ $\bar{p}_l \leftarrow \bar{p}_l(B_2)$ via \eqref{eq:pbar_B2}, $\forall l \in \mathcal{A}$
    \STATE Re-compute the estimates and combiners at the updated $B_1$; evaluate $\mathsf{EE}^{(n)}$ \ (every step above is accepted under the safeguard rule)
\UNTIL{$\big( \mathsf{EE}^{(n)} - \mathsf{EE}^{(n-1)} \big) / \mathsf{EE}^{(n-1)} \leq \epsilon$}
\ENSURE $(t_1, t_2, B_1, B_2, \{b_l\}, \{\bar{p}_l\}, \mathcal{A})$
\end{algorithmic}
\end{algorithm}

\begin{proposition}
\label{prop:convergence}
The sequence $\{\mathsf{EE}^{(n)}\}$ generated by Algorithm~\ref{alg:ao} is nondecreasing and convergent.
\end{proposition}
\begin{IEEEproof}
Every block update is subject to the safeguarded acceptance rule, so no step decreases the EE: Step 1 maximizes the EE exactly over its block (Proposition~\ref{prop:tsplit}); Step 2 is accepted only if it improves the (frozen-statistics) EE; Step 3 selects the best candidate from a set that contains the incumbent, so it cannot decrease the EE; and Step 4 keeps the throughput unchanged and does not increase the consumed energy (Lemma~\ref{lem:B2convex} and the subsequent projection), so it cannot decrease the EE either. Any switch-off is part of the candidate selected in Step 3, whose EE is at least that of the incumbent, so switch-off events do not break the monotonicity; moreover, since the active set is nonincreasing along the iterations, at most $L - 1$ such events can occur. The sequence $\{\mathsf{EE}^{(n)}\}$ is thus nondecreasing; it is bounded above because the throughput is bounded (finite bandwidth and transmit powers) and $P_{\rm tot} \geq P^{\rm fix} > 0$. A monotone bounded sequence converges.
\end{IEEEproof}

\section{Bussgang-Based Monte Carlo Performance Evaluation}
\label{sec:bussgang}

The AQNM adopted in Sections~\ref{sec:model} and \ref{sec:algorithm} treats the quantization distortion as an additive, input-independent, and temporally uncorrelated impairment. This renders the joint resource allocation tractable, but it only approximates the true quantizer behavior, since the Bussgang gain and the distortion covariance of an actual scalar quantizer do not obey the idealized AQNM statistics \cite{demir2020bussgang}. Therefore, all performance results are obtained by evaluating the configuration $(t_1, t_2, B_1, B_2, \{b_l\}, \{\bar{p}_l\}, \mathcal{A})$ returned by Algorithm~\ref{alg:ao} in an end-to-end manner, where the \emph{exact} Lloyd--Max quantizers are applied to both the pilot and the data samples. The channel estimation and the receive combining are kept unchanged from the AQNM-based design of Sections~\ref{sec:model} and \ref{sec:algorithm}, whereas the achievable rate is computed via the end-to-end Bussgang decomposition of the actually-quantized data path. We emphasize that the Bussgang-based characterization of this section is used \emph{only} for computing the reported rates; the AQNM is employed within the optimization in Algorithm~\ref{alg:ao} as well as in the estimation and combining operations of the evaluated transceiver.

\subsection{Pilot Quantization and AQNM-Based Estimation and Combining}
\label{subsec:bussgang_pilot}

In each Monte Carlo realization, the despread pilot observation of AP $l \in \mathcal{A}$ associated with the pilot of UE $k$, $\vect{z}_{lk} = \sum_{i \in \mathcal{P}_k} \sqrt{p_i \tau_p} \vect{h}_{li} + \tilde{\vect{n}}_{lk}$, is passed through the \emph{actual} Lloyd--Max quantizer with the fixed pilot resolution $b_{\rm p}$, applied separately to the real and imaginary parts of each entry, yielding $\tilde{\vect{z}}_{lk,{\rm Q}} = \mathcal{Q}_l^{\rm p}( \vect{z}_{lk} )$.

Importantly, the channel estimation and the receive combining are kept \emph{exactly as in the AQNM-based design}: the CPU applies the quantization-aware LMMSE estimator of Section~\ref{subsec:channel_estimation} to the actually-quantized observation, i.e.,
\begin{equation}
    \hat{\vect{h}}_{lk} = \sqrt{p_k \tau_p}\, \vect{R}_{lk}\, \bm{\Psi}_{lk}^{-1}\, \tilde{\vect{z}}_{lk,{\rm Q}},
    \label{eq:aqnm_estimate_eval}
\end{equation}
with $\bm{\Psi}_{lk}$ computed from the AQNM statistics as in \eqref{eq:lmmse}, and the receive combiners $\{\vect{v}_k\}$ are obtained from these estimates via the quantization-aware MMSE rule \eqref{eq:mmse_combiner} with the AQNM distortion covariances, exactly as within Algorithm~\ref{alg:ao}. This choice is deliberately model-mismatched: the CPU operates with the tractable AQNM statistics, while the pilot and data samples have traversed the true nonlinear quantizers. The evaluation thus measures the end-to-end performance that the AQNM-designed transceiver \emph{actually delivers}, rather than the performance of an idealized Bussgang-optimal receiver.

\subsection{End-to-End Bussgang Model for the Quantized Data Path}
\label{subsec:bussgang_data}

During the data transmission, the exact scalar quantizer of AP $l \in \mathcal{A}$, operating at the optimized resolution $b_l$, is applied to $\vect{y}_l$, and the CPU descales the output by $1/(1-\eta_l)$ as in Section~\ref{subsec:quantization}; stacking over the active APs yields the observation $\tilde{\vect{y}}_{\rm B}$, a deterministic nonlinear function of $\vect{y} = \sum_{i=1}^{K} \sqrt{p_i}\, \vect{h}_i s_i + \vect{n}$. We form the Bussgang decomposition of the \emph{concatenated} observation directly with respect to the symbol vector $\vect{s} = [s_1, \ldots, s_K]^{\Ttran}$:
\begin{equation}
    \tilde{\vect{y}}_{\rm B} = \vect{F} \vect{s} + \vect{d}, \qquad
    \vect{F} = \E\big\{ \tilde{\vect{y}}_{\rm B}\, \vect{s}^{\Htran} \,\big|\, \vect{H} \big\},
    \label{eq:bussgang_observation}
\end{equation}
where $\vect{F} \in \mathbb{C}^{L_{\mathcal{A}} N \times K}$ is the \emph{end-to-end effective channel matrix}, subsuming the transmit powers, the true channels, and the per-AP quantizer gains. The residual $\vect{d} = \tilde{\vect{y}}_{\rm B} - \vect{F} \vect{s}$ satisfies $\E\{ \vect{d}\, \vect{s}^{\Htran} \,|\, \vect{H} \} = \vect{0}$ \emph{by construction}. The residual collects the quantizer-processed thermal noise \emph{and} the quantization distortions of \emph{all} active APs, with the conditional covariance
\begin{equation}
    \vect{C}_{\vect{d}} = \E\big\{ \tilde{\vect{y}}_{\rm B} \tilde{\vect{y}}_{\rm B}^{\Htran} \,\big|\, \vect{H} \big\} - \vect{F} \vect{F}^{\Htran},
    \label{eq:Cd}
\end{equation}
which retains all inter-AP distortion cross-correlations without any block-diagonal approximation. Since $\vect{F}$ and $\vect{C}_{\vect{d}}$ do not admit closed-form expressions for arbitrary quantizer levels, they are estimated empirically via sample averaging with the actual Lloyd--Max quantizers.

\subsection{Bussgang-Based Ergodic Spectral and Energy Efficiency}
\label{subsec:bussgang_rate}

Using the AQNM-based channel estimates and combiners of Section~\ref{subsec:bussgang_pilot}, the CPU detects $\hat{s}_k = \vect{v}_k^{\Htran} \tilde{\vect{y}}_{\rm B}$, where $\tilde{\vect{y}}_{\rm B}$ is the actually-quantized observation in \eqref{eq:bussgang_observation}. Following the use-and-then-forget (UatF) bounding technique \cite{demir2021foundations}, UE $k$ relies only on the average effective channel $\mu_k = \E\{ \vect{v}_k^{\Htran} \vect{f}_k \}$ for coherent detection, where $\vect{f}_k$ denotes the $k$th column of $\vect{F}$, while all the remaining randomness is treated as uncorrelated effective noise. An achievable ergodic SINR of UE $k$ is then
\begin{equation}
    \gamma_{k,{\rm B}} = \frac{\left| \mu_k \right|^2}{\mathsf{I}_k + \mathsf{DN}_k},
    \label{eq:bussgang_sinr}
\end{equation}
where
\begin{align}
    \mathsf{I}_k &= \sum_{i=1}^{K} \E\left\{ \left| \vect{v}_k^{\Htran} \vect{f}_i \right|^2 \right\} - \left| \mu_k \right|^2,
    \label{eq:bussgang_interf} \\
    \mathsf{DN}_k &= \E\left\{ \vect{v}_k^{\Htran} \vect{C}_{\vect{d}}\, \vect{v}_k \right\}
    \label{eq:bussgang_noise}
\end{align}
denote the beamforming uncertainty plus multi-user interference and the residual distortion-plus-noise power, respectively; the latter includes all inter-AP distortion cross-correlations through \eqref{eq:Cd}. Note that the transmit powers are absorbed into $\vect{F}$, so no explicit power factor appears in \eqref{eq:bussgang_sinr}. The corresponding achievable ergodic SE and effective throughput are
\begin{equation}
    \mathsf{SE}_{k,{\rm B}} = \frac{\tau_u}{\tau_c} \log_2\left( 1 + \gamma_{k,{\rm B}} \right), \qquad
    R_{k,{\rm B}} = t_1 B_1\, \mathsf{SE}_{k,{\rm B}},
    \label{eq:bussgang_rate}
\end{equation}
and the reported EE is $\mathsf{EE}_{\rm B} = \sum_{k=1}^{K} R_{k,{\rm B}} / P_{\rm tot}$, with $P_{\rm tot}$ from \eqref{eq:total_power} evaluated at the configuration returned by Algorithm~\ref{alg:ao}.

The expectations in \eqref{eq:bussgang_sinr}--\eqref{eq:bussgang_noise} are computed via Monte Carlo averaging over independent channel realizations: in each realization, the actual Lloyd--Max quantizers are applied to the despread pilot signals and to the received data samples of all active APs, the channel estimates and combiners are formed by the AQNM-based rules of Section~\ref{subsec:bussgang_pilot}, the end-to-end effective channel $\vect{F}$ and the residual covariance $\vect{C}_{\vect{d}}$ are estimated empirically, and the per-UE effective-channel and residual powers are accumulated. Since the rate expression itself does not rely on the AQNM, the reported EE accounts for the input dependence and correlations of the true quantization errors, while reflecting the exact receiver processing employed by the proposed design.

\vspace{-2mm}
\section{Numerical Results}
\label{sec:results}

This section evaluates the EE gains of the proposed methods relative to benchmark schemes that do not optimize the bandwidth, power, or time allocations. We begin by describing the simulation setup and then present the corresponding numerical results.
\subsection{Simulation Setup}
\label{subsec:setup}

We consider a cell-free massive MIMO network deployed over a $1\,\mathrm{km}\times1\,\mathrm{km}$ square area, with the CPU located at the center. Unless otherwise specified, $L=16$ APs, each equipped with $N=4$ antennas, jointly serve $K=10$ single-antenna UEs, all uniformly distributed over the coverage area. The APs are mounted at a height of $10$\,m, while the CPU is located at a height of $20$\,m.
The access and wireless fronthaul links
operate in the upper mid-band at a carrier frequency of $f_c=7.5$\,GHz and
share a total bandwidth of $B=500$\,MHz; the noise power spectral density
follows from a $-174$\,dBm/Hz thermal floor with a $5$\,dB noise figure. Each
UE transmits with $p_k = 0.2$\,W, the coherence block spans $\tau_c=200$
samples with $\tau_p=K$ orthogonal pilots, and the pilots are quantized with a fixed resolution of
$b_{\rm p}=4$\,bits while the per-AP data resolution is optimized in
$\{0,1,\dots,b_{\max}\}$ with $b_{\max}=12$\,bits per real dimension. The
wireless fronthaul is received by a CPU uniform circular array (UCA) with
$M_c=256$ antennas, and the per-AP fronthaul power budget is
$P_{\mathrm{fh}}^{\max}=10$\,W.

The large-scale fading of the \emph{access} links follows a 3GPP Urban Microcell
non-line-of-sight (NLOS) model \cite[Table 7.4.1-1]{3gpp38901},
\begin{equation}
\beta_{lk}\,[\mathrm{dB}] = -32.4 - 31.9\log_{10}(d_{lk})
                            - 20\log_{10}(f_c) + F_{lk},
\end{equation}
where $d_{lk}$ is the three-dimensional AP--UE distance in meters, $f_c$ is in
GHz, and $F_{lk}\sim\mathcal{N}(0,8.2^2)$ is the shadow fading in dB; the
resulting channels are spatially uncorrelated Rayleigh fading. The
\emph{wireless fronthaul} links are line-of-sight (LOS) with \cite[Table 7.4.1-1]{3gpp38901}
\begin{equation}
\beta^{\mathrm{fh}}_{l}\,[\mathrm{dB}] = -32.4 - 21\log_{10}(d_{l})
                             - 20\log_{10}(f_c) + F^{\mathrm{fh}}_{l},
\end{equation}
with $d_l$ the AP--CPU distance and $F^{\mathrm{fh}}_{l}\sim\mathcal{N}(0,4^2)$;
the effective CPU-side channel of AP $l$ is
$\mathbf{g}_l=\sqrt{N\beta^{\mathrm{fh}}_l}\,\mathbf{a}(\varphi_l,\theta_l)$, where the
factor $N$ is the AP transmit-beamforming array gain and $\mathbf{a}(\cdot)$ is
the UCA steering vector. The power-model parameters are summarized in
Table~\ref{tab:params}.

The reported metric is the network EE in Mbit/Joule,
\emph{evaluated end-to-end with the actual Lloyd--Max quantizers and the
Bussgang bound of Section~\ref{subsec:bussgang_rate}} at the configuration
returned at convergence (using $T_{\mathrm{eval}}=250$ channel realizations
and $n_{\mathrm{MC}}=500$ data symbols per realization). Every point is
averaged over $100$ independent AP/UE drops; setups in which a scheme cannot
support even $b_l=1$ on any AP (all-off, i.e., infeasible) are counted as
$\mathrm{EE}=0$ in the average.

We compare four schemes:
\begin{itemize}
  \item \textbf{Proposed, TD mode} and \textbf{Proposed, FD mode}: Algorithm~1
  run in the TD and FD modes, respectively, jointly
  optimizing the access/fronthaul split, the per-link bandwidths, the per-AP
  quantization resolutions (with sleep-based AP switch-off), and the fronthaul
  powers.
  \item \textbf{Benchmark~1 (TD-fixed)}: an equal, non-optimized split
  $t_1=t_2=\tfrac12$, $B_1=B_2=B$, with full fronthaul power and the
  resolutions set by the fronthaul-constraint-saturating rule (with sleep if
  even $b_l=1$ is infeasible).
  \item \textbf{Benchmark~2 (FD-fixed)}: a fixed split $t_1=t_2=1$,
  $B_1=B_2=B/2$, otherwise identical to Benchmark~1.
\end{itemize}
Since each proposed scheme is \emph{initialized} at its corresponding
fixed-resource benchmark and Algorithm~1 accepts only EE-nondecreasing
updates, the proposed TD and FD schemes are expected to upper-bound
Benchmarks~1 and~2, respectively, throughout.

\begin{table}[t]
\centering
\caption{Default simulation and power-model parameters.}
\vspace{-2mm}
\label{tab:params}
\begin{tabular}{llll}
\hline
Parameter & Value & Parameter & Value \\
\hline
$L$ / $N$ / $K$        & $16$ / $4$ / $10$ & $\mu_{\mathrm{AP}}$        & $0.1$\,W \\
$M_c$                  & $256$             & $D_0$                     & $0.1$\,W \\
Area                   & $1\times1$\,km$^2$ & $\nu$                    & $10^{-10}$\,W/Hz \\
$B$                    & $500$\,MHz        & $\kappa_{\mathrm{UE}}$     & $0.4$ \\
$f_c$                  & $7.5$\,GHz        & $P_{0,\mathrm{UE}}$        & $0.1$\,W \\
Noise figure           & $5$\,dB           & $\kappa_{\mathrm{fh}}$     & $0.4$ \\
$p_k$                  & $0.2$\,W          & $P_{0,\mathrm{fh}}$        & $2$\,W \\
$\tau_c$ / $\tau_p$    & $200$ / $K$       & $\mu_{\mathrm{CPU}}^{\rm fh}$    & $0.1$\,W \\
$b_{\rm p}$ / $b_{\max}$     & $4$ / $12$\,bits  & $D_{0}^{\rm c}$ / $\nu_c$        & $0.1$\,W / $10^{-10}$\,W/Hz \\
$P_{\mathrm{fh}}^{\max}$ & $10$\,W         & $P_{\mathrm{CPU}}$         & $50$\,W \\
$T_{\mathrm{eval}}$    & $250$             & $\eta_{\mathrm{dec}}$      & $10^{-9}$\,W$\cdot$s/bit \\
$n_{\mathrm{MC}}$      & $500$             & $\delta_{\rm s}$ (sleep depth)    & $0.3$ \\
\hline
\end{tabular}
\vspace{-2mm}
\end{table}

\subsection{Results and Discussion}
\label{subsec:discussion}

\begin{figure}[t]
\centering
 \includegraphics[width=\columnwidth]{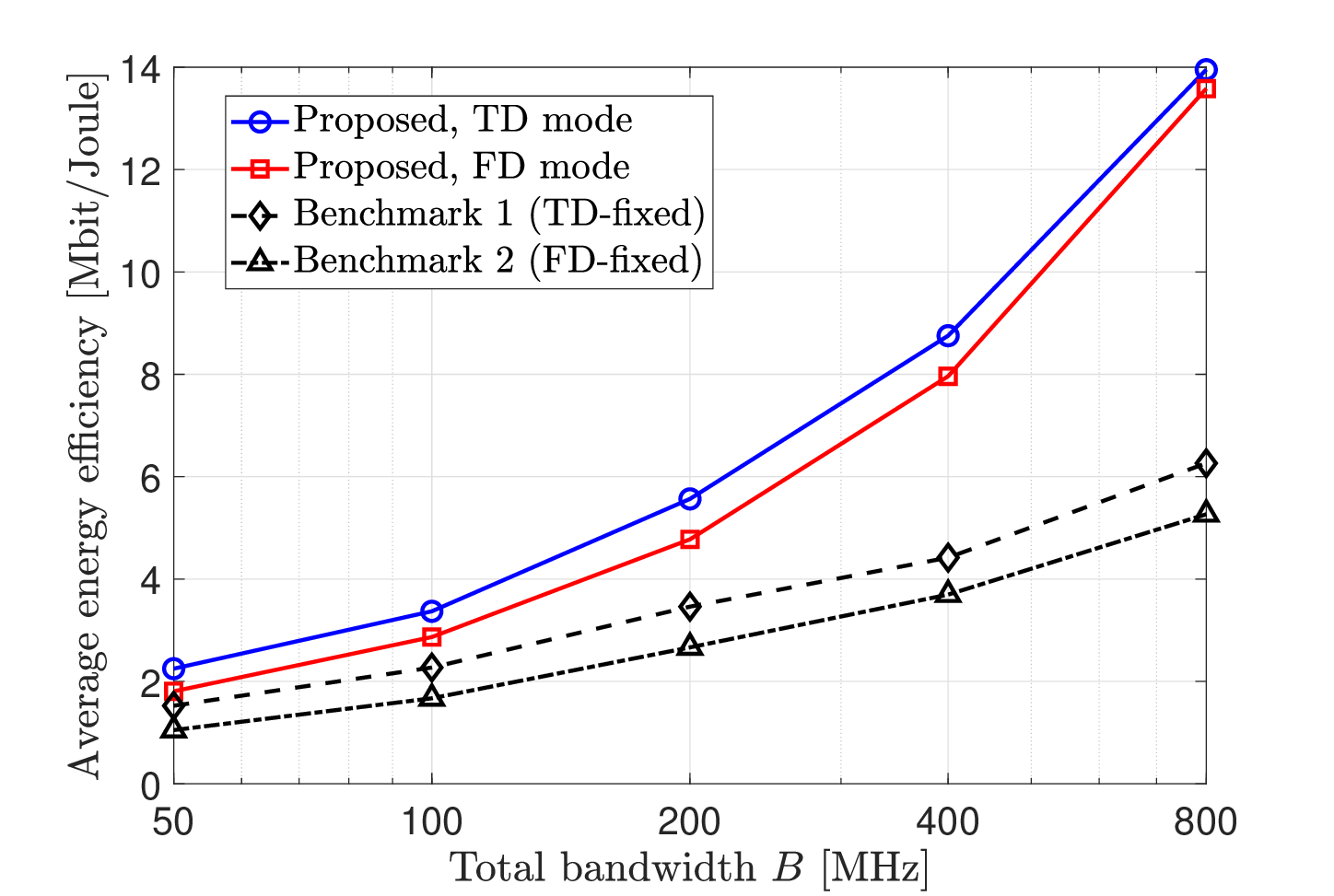}   
 \vspace{-6mm}
\caption{Average EE versus the total bandwidth $B$.}
\label{fig:ee_bw}
\vspace{-4mm}
\end{figure}

\textbf{Impact of the total bandwidth (Fig.~\ref{fig:ee_bw}).}
Fig.~\ref{fig:ee_bw} shows the average EE as the total bandwidth $B$ is
swept from $50$ to $800$\,MHz. For all schemes, the EE increases with $B$, as
the additional bandwidth boosts the access pre-log more than it inflates the
noise and the bandwidth-proportional power at the considered operating points.
More importantly, the two proposed schemes increasingly outperform the
fixed-split benchmarks as $B$ grows: with a larger bandwidth to apportion, the
freedom to jointly adapt the access/fronthaul split and the per-AP resolutions
becomes more valuable, so the gain of the optimization is most pronounced in
the wideband regime. Among the two modes, the TD scheme attains the highest
EE. This is because time division lets the access and fronthaul interfaces
sleep during each other's active phase, so the sleep modes of \emph{both}
links can be fully exploited. In the FD mode, this benefit is only partially
realized: since the band is split between the two simultaneously active links,
the optimization tends to keep the activity fractions (and hence the operating
times $t_1,t_2$) as large as possible in order to preserve the rate, which
leaves little idle time for sleep. Consequently, Proposed-FD lies between
Proposed-TD and the benchmarks, and the same TD-over-FD ordering carries over
to the fixed schemes (Benchmark~1 above Benchmark~2).

\textbf{Impact of the number of UEs (Fig.~\ref{fig:ee_K}).}
Fig.~\ref{fig:ee_K} reports the average EE versus the number of UEs $K$. The
EE increases steadily with $K$ for all schemes: adding UEs strengthens the
spatial-multiplexing gain, and the resulting growth of the sum throughput
outweighs the accompanying increase in pilot overhead, multi-user
interference, and fronthaul load. Remarkably, this trend persists up to
$K=16$, where the network serves $16$ UEs with only $M=LN=64$ total AP
antennas---an appreciably loaded massive MIMO regime---yet the multiplexing
gain still dominates and the EE keeps rising. The proposed schemes retain the
highest EE across the whole range, with the TD mode leading throughout.
Moreover, the advantage of TD over FD \emph{widens} as $K$ increases: a larger
number of UEs raises the aggregate data rate and thus makes the fronthaul more
critical, so TD's ability to use the full band in each phase while letting the
idle interface sleep becomes increasingly valuable, whereas the band-splitting
FD mode---which keeps both links active over long times---is progressively more
constrained.

\begin{figure}[t]
\centering
\includegraphics[width=\columnwidth]{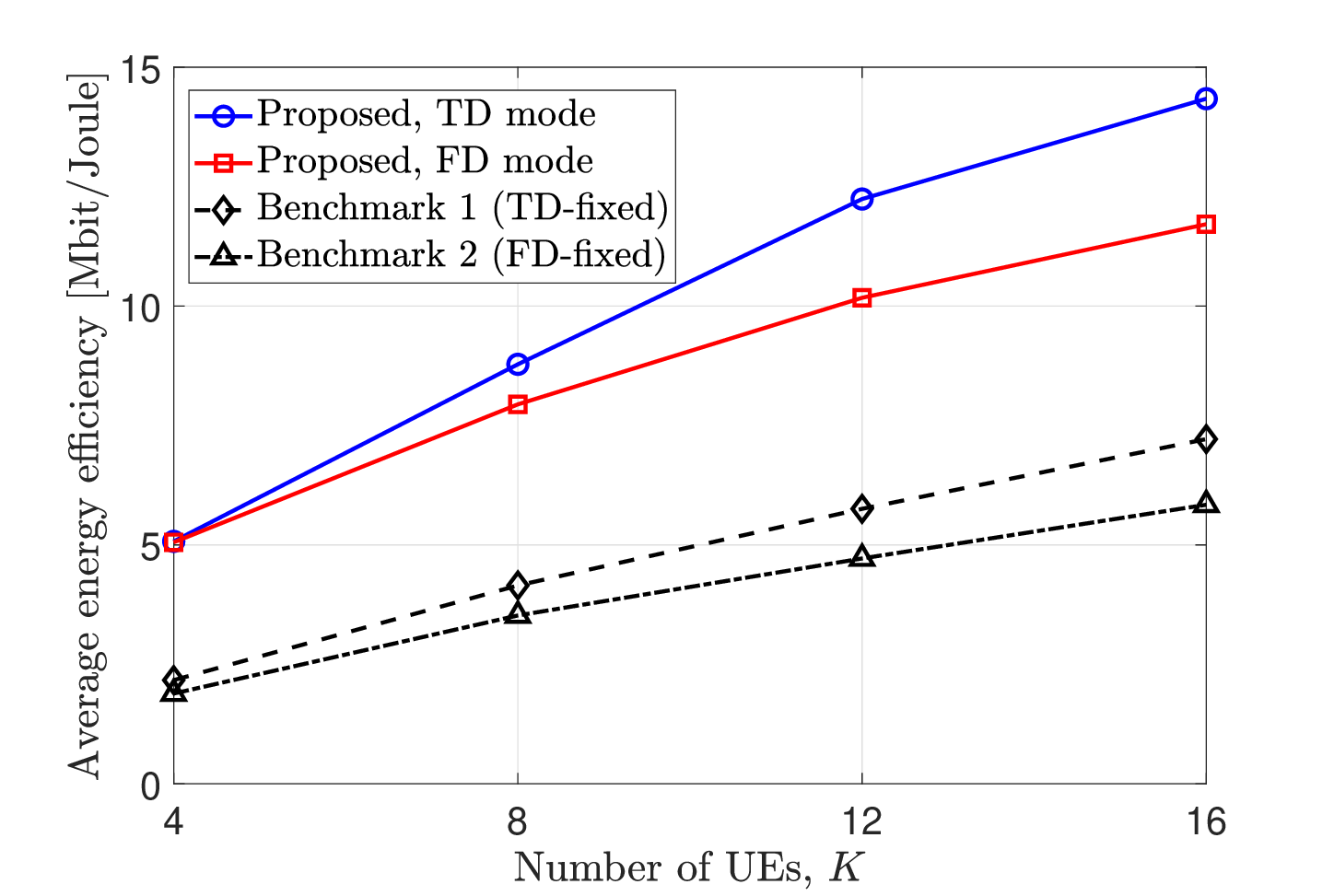}   
\vspace{-6mm}
\caption{Average EE versus the number of UEs $K$.}
\label{fig:ee_K}
\vspace{-4mm}
\end{figure}

\textbf{Impact of the CPU fronthaul array size (Fig.~\ref{fig:ee_Mc}).}
Fig.~\ref{fig:ee_Mc} varies the number of CPU fronthaul antennas $M_c$. For
the proposed schemes the EE rises steeply as $M_c$ increases from $16$ to $32$
and reaches a maximum at an intermediate value (around $M_c=64$--$128$), after
which it declines, revealing an EE-maximizing operating point. This
non-monotonic behavior reflects a clear trade-off: the $M_c$ CPU antennas are
always active, so each of them draws static and bandwidth-proportional power
regardless of the load. A larger array improves the fronthaul channel gain and
the conditioning of the zero-forcing fronthaul inverse, which is highly
beneficial while the fronthaul is the bottleneck; once the fronthaul is no
longer limiting, however, the extra rate it provides exceeds what is needed,
and its growing power cost drives the EE down. This peak is essentially absent
for the fixed benchmarks, whose EE instead increases slowly and monotonically
with $M_c$. A likely reason is that the benchmarks cannot shape the
fronthaul-rate demand through adaptive time and bandwidth allocation; for them,
enlarging $M_c$ mainly improves feasibility---turning setups that would
otherwise be infeasible (and counted as $\mathrm{EE}=0$) into feasible
ones---so their average EE keeps edging upward rather than peaking and falling.
Finally, at small $M_c$ the FD mode slightly outperforms the TD mode before TD
takes over for larger arrays.

\begin{figure}[t]
\centering
\includegraphics[width=\columnwidth]{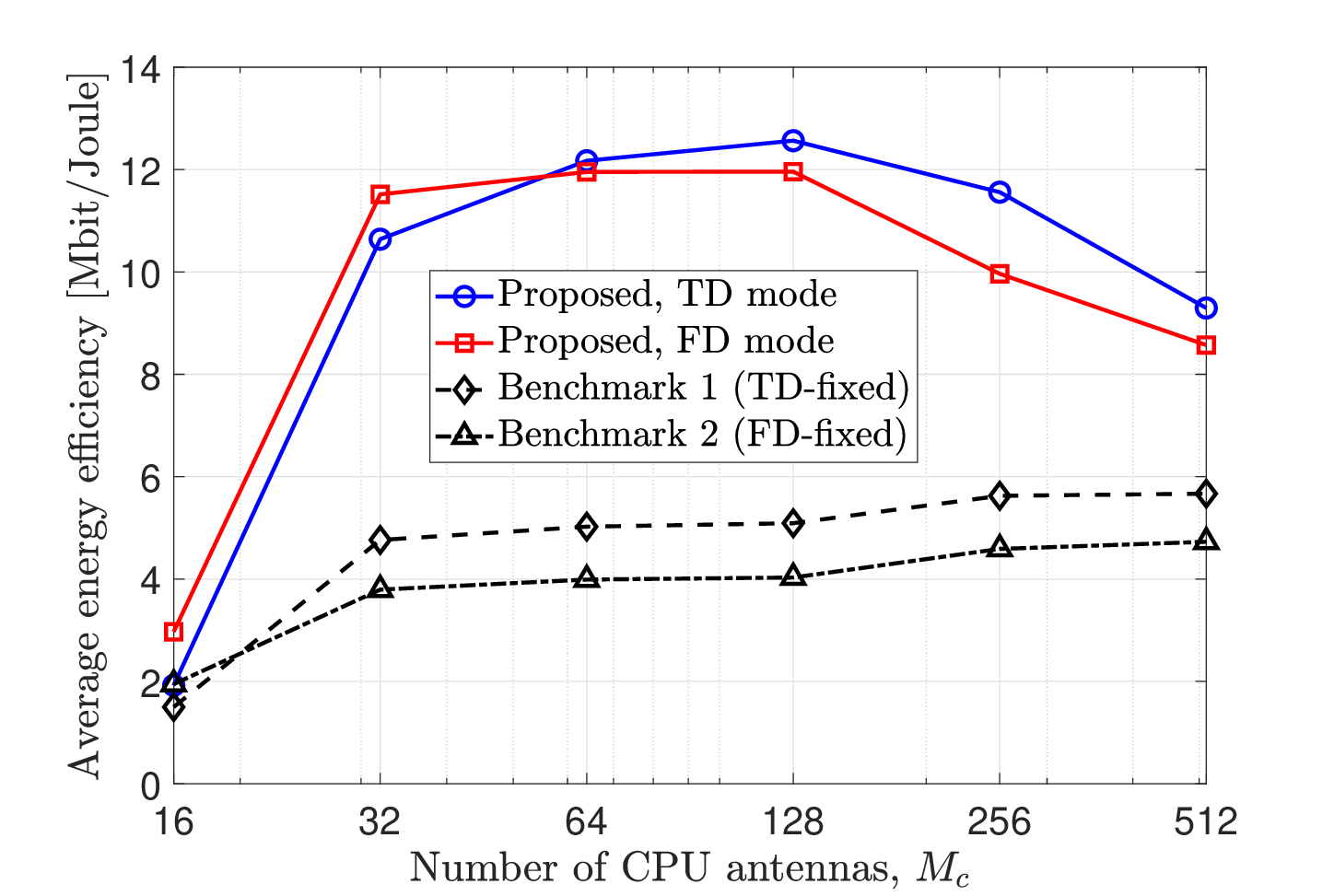}  
\vspace{-6mm}
\caption{Average EE versus the number of CPU fronthaul
antennas $M_c$.}
\label{fig:ee_Mc}
\vspace{-4mm}
\end{figure}

\textbf{Impact of the number of APs (Fig.~\ref{fig:ee_L}).}
Fig.~\ref{fig:ee_L} shows the average EE as the number of APs $L$ grows from
$8$ to $64$ (with $M_c=256$ fixed). The two proposed schemes improve markedly
with $L$: a denser AP layer offers more macro-diversity and shorter access
distances, and---crucially---the sleep-based switch-off gives the optimization
a larger pool of APs from which to activate only the most useful ones while
putting the rest to sleep. Additional APs thus act as a diversity resource for
the proposed design, steadily raising the EE. In stark contrast, the fixed
benchmarks remain almost flat and even decline slightly at $L=64$: since they
cannot adapt the fronthaul load through the time/bandwidth allocation and keep
all supportable APs active at full fronthaul power without energy-driven
sleeping, they cannot effectively harvest the macro-diversity of the extra
APs, whose added circuit and fronthaul power instead offsets the throughput
gain. This contrast underscores that the value of a dense AP deployment is
unlocked only when the resource split and the per-AP activation are jointly
optimized.

\textbf{Impact of the antennas per AP (Fig.~\ref{fig:ee_N}).}
Fig.~\ref{fig:ee_N} varies the number of antennas per AP $N\in\{1,2,4,8\}$
and again reveals an EE-optimal ``sweet spot'' (around $N=4$ for the proposed
schemes). Two opposing effects shape this curve. On the rising side,
increasing $N$ enlarges both the fronthaul transmit-beamforming gain (which
scales with $N$) and the access combining gain, so it improves the fronthaul
and access rates alike and lifts the EE---even though a smaller $N$ would allow
more bits per real dimension at the same fronthaul rate, the array-gain benefit
dominates in this region. Beyond the sweet spot, however, the EE drops sharply:
under split option~8 each AP must quantize $N$ times as many real dimensions,
so at the fixed per-AP fronthaul budget the affordable per-dimension resolution
collapses; in addition, feasibility becomes jeopardized, as more APs can no
longer support even $b_l=1$ and are forced into sleep (contributing
$\mathrm{EE}=0$). These two effects---a steep loss of allocatable bits and a
loss of feasibility---together produce the pronounced decline at $N=8$.
Interestingly, the FD mode is markedly more resilient to this drop than TD,
retaining a substantially higher EE at $N=8$.

\begin{figure}[t]
\centering
\includegraphics[width=\columnwidth]{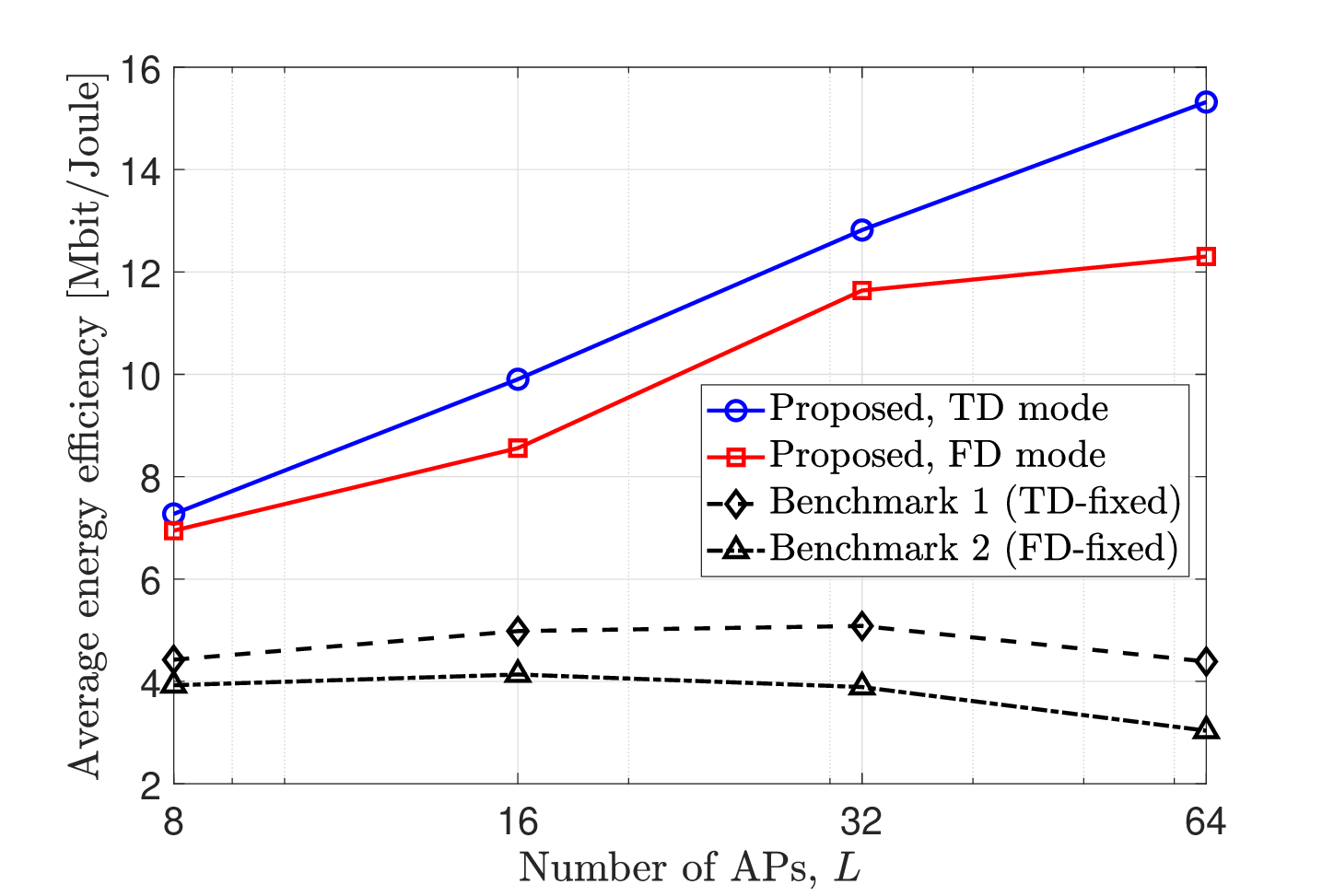}   
\vspace{-6mm}
\caption{Average EE versus the number of APs $L$.}
\label{fig:ee_L}
\vspace{-4mm}
\end{figure}

\begin{figure}[t]
\centering
\includegraphics[width=\columnwidth]{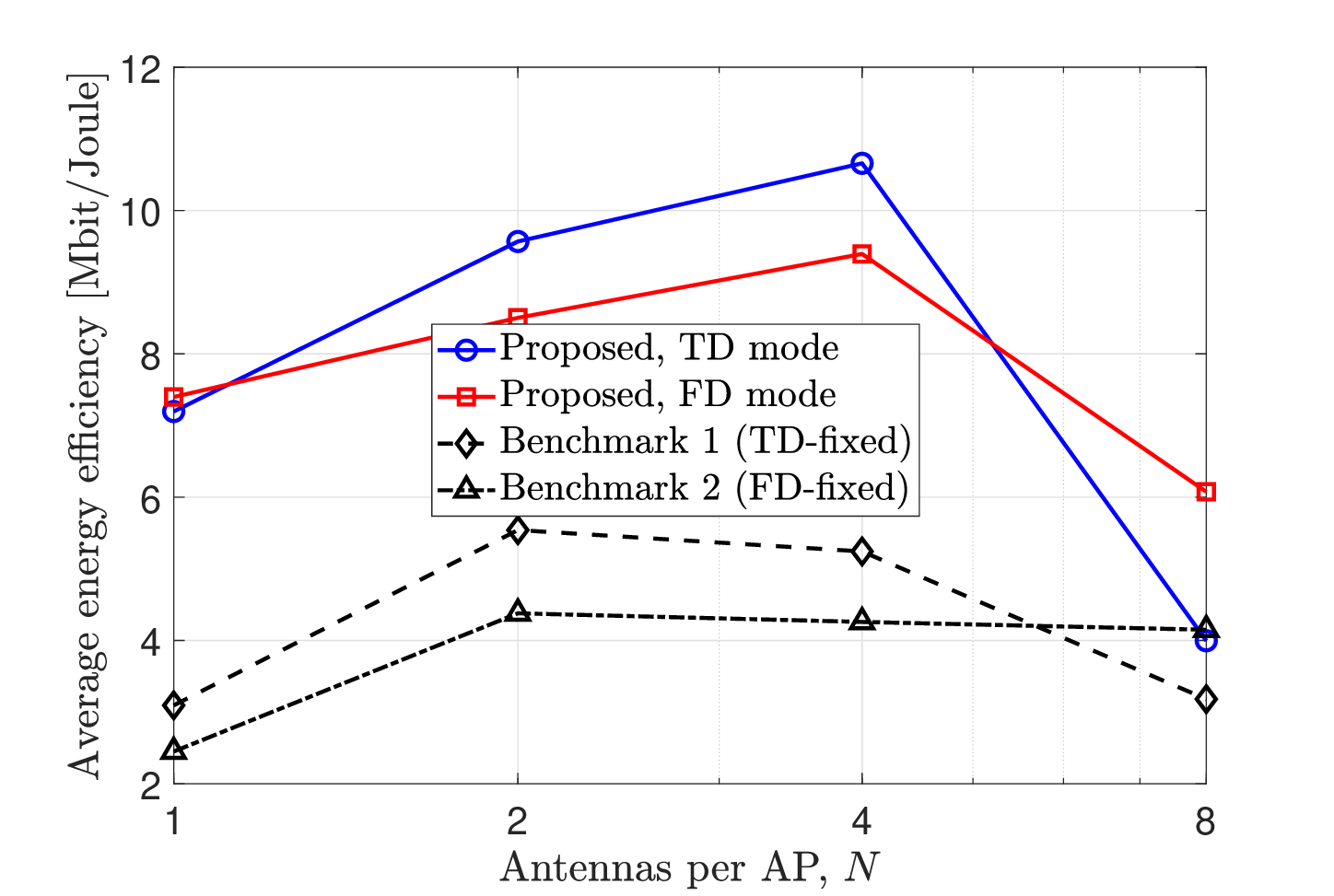}   
\vspace{-6mm}
\caption{Average EE versus the number of antennas per AP $N$.}
\label{fig:ee_N}
\vspace{-4mm}
\end{figure}

\section{Conclusion}
\label{sec:conclusion}
We studied the energy-efficient uplink operation of a cell-free massive MIMO
network with wireless fronthaul, in which the access and fronthaul links share
a common band under an IAF design. We
formulated a network EE-maximization problem that jointly optimizes the
access/fronthaul time split, the per-link bandwidths, the adaptive per-AP
quantization resolutions with sleep-based AP switch-off, and the fronthaul
powers, unified across the TD and FD modes, and we solved the resulting
mixed-integer, nonconvex fractional program with a safeguarded
alternating-optimization algorithm whose per-block updates admit closed-form or
bisection solutions. Validated end-to-end with the actual Lloyd--Max quantizers
and a Bussgang bound over a wideband upper-mid-band deployment, the proposed
schemes substantially outperformed fixed-resource benchmarks. The results
showed that the TD mode is generally the most energy efficient---since time
division lets both the access and fronthaul interfaces exploit their sleep
modes---and that operating the APs at low quantization resolutions together
with selective AP sleep is central to energy efficiency; they also revealed
EE-optimal operating points in the CPU fronthaul array size and in the number
of antennas per AP.

A natural extension of this work is a \emph{space-division} mode for the IAF
problem: rather than separating the access and fronthaul links in time (TD) or
in frequency (FD), SDMA would let both links
operate simultaneously over the full band while being separated in the spatial
domain by the AP and CPU arrays. Jointly designing such spatial separation with
the adaptive quantization and resource allocation---along with a downlink
counterpart and other multiple-access schemes---is a promising direction for
future research.

\section*{Acknowledgment}

During the preparation of this work, the author used ChatGPT (OpenAI, GPT-5)
and Claude (Anthropic, Claude Opus 4.8) to assist with improving the clarity of
LaTeX formatting, language polishing, and manuscript organization. Example
prompts included requests for polishing technical explanations, improving the
presentation of mathematical derivations, and refining the structure of
specific manuscript sections. After using these tools, the author carefully
reviewed, edited, and verified all generated content and takes full
responsibility for the content of the publication.

\bibliographystyle{IEEEtran}
\bibliography{IEEEabrv,refs,refs2}
\end{document}